\documentclass[10pt,twocolumn,twoside]{IEEEtran}

\usepackage{cite}
\usepackage[cmex10]{amsmath}
\usepackage{algorithmic}
\usepackage{array}
\usepackage[tight,footnotesize]{subfigure}
\usepackage{url}
\usepackage{amsthm}
\usepackage{amssymb}
\usepackage{mathrsfs} 
\usepackage{graphicx}
\usepackage{multirow}
\usepackage{rotating}
\usepackage{color}
\usepackage{dsfont}

\newcommand{\abs}[1]{|#1|}

\newcommand{\y}{\mathbf{y}}
\newcommand{\n}{\mathbf{n}}
\newcommand{\h}{\mathbf{H}}
\newcommand{\sbf}{\mathbf{s}}
\newcommand{\Lbf}{\mathbf{L}}
\newcommand{\ubf}{\mathbf{u}}
\newcommand{\zbf}{\mathbf{z}}
\newcommand{\ybf}{\mathbf{y}}
\newcommand{\aalpha}{\boldsymbol{\alpha}}

\newcommand{\Lop}{{\rm L}}

\def\half{{\textstyle\frac{1}{2}}}
\def\dint{\;\mathrm{d}}

\newcommand{\inR}{\in \mathbb{R}}

\newcommand{\R}{\mathbb{R}}
\newcommand{\C}{\mathbb{C}}
\newcommand{\N}{\mathbb{N}}
\newcommand{\Z}{\mathbb{Z}}

\newcommand{\Dop}{{\rm D}}
\newcommand{\bx}{{\boldsymbol x}}
\newcommand{\bk}{{\boldsymbol k}}
\def\CH{{\widehat{\mathscr{P}}}}

\newtheorem{theo}{Theorem}

\newtheorem{defn}{Definition}

\begin{document}

%%%%%%%%%%%%%%%%%%%%%%%%%%%%%%%%%%%%%%%%%%%%%%%%%%%%%%%%%%%%%%%%%

\title{Sparse Stochastic Processes and Discretization of Linear Inverse Problems}
\author{Emrah~Bostan$^*$,~\IEEEmembership{Student~Member,~IEEE,}
        Ulugbek~S.~Kamilov,~\IEEEmembership{Student~Member,~IEEE,}
        	Masih~Nilchian,~\IEEEmembership{Student~Member,~IEEE,}
        and~Michael~Unser,~\IEEEmembership{Fellow,~IEEE}
\thanks{This work was partially supported by the Center for Biomedical Imaging of the Geneva-Lausanne Universities and EPFL, as well as by the foundations Leenaards and Louis-Jeannet and by the European Commission under Grant ERC-2010-AdG 267439-FUN-SP.}
\thanks{The authors are with the Biomedical Imaging Group, \'Ecole polytechnique f\'ed\'erale de Lausanne (EPFL), Station 17, CH--1015 Lausanne VD, Switzerland}}

\markboth{Sparse Stochastic Processes and Discretization of Linear Inverse Problems}
{Bostan, Kamilov, Nilchian, and Unser}

\maketitle

\ifCLASSOPTIONpeerreview
\begin{center}
 \vspace{-0.6cm}
 Authors' contact information:\\
 \vspace{0.2cm}
 \'{E}cole polytechnique f\'ed\'erale de Lausanne\\
 Biomedical Imaging Group\\
 CH-1015 Lausanne VD, Switzerland.\\
 Tel: +41(0)216935136, Fax: +41(0)216933701.\\
 Email: \url{{emrah.bostan, ulugbek.kamilov, masih.nilchian, michael.unser}@epfl.ch}\\
 Web: \url{http://bigwww.epfl.ch/}\\
 \vspace{0.2cm}
\end{center}
\fi

\begin{abstract}
\boldmath 
We present a novel statistically-based discretization paradigm and derive a class
of maximum a posteriori (MAP) estimators for solving ill-conditioned linear inverse
problems. We are guided by the theory of sparse stochastic processes, which specifies continuous-domain signals as
solutions of linear stochastic differential equations. Accordingly, we show that the
class of admissible priors for the discretized version of the signal is
confined to the family of infinitely divisible distributions. 
Our estimators not only cover the well-studied methods of Tikhonov and
$\ell_1$-type regularizations as particular cases, but also open the
door to a broader class of sparsity-promoting regularization schemes that are typically nonconvex. We provide an algorithm that  handles the corresponding nonconvex problems and illustrate the use of our formalism by applying it to deconvolution, MRI, and X-ray tomographic reconstruction problems. Finally, we compare the performance of estimators associated with models of increasing sparsity.
\end{abstract}

\ifCLASSOPTIONjournal
\fi

\ifCLASSOPTIONpeerreview
 \begin{center} \bfseries EDICS Category: SAS-STAT,  IMD-ANAL\end{center}
\fi

\IEEEpeerreviewmaketitle

%%%%%%%%%%%%%%%%%%%%%%%%%%%%%%%%%%%%%%%%%%%%%%%%%%%%%%%%%%%%%%%%%
\section{Introduction}\label{sec:Introduction}
%%%%%%%%%%%%%%%%%%%%%%%%%%%%%%%%%%%%%%%%%%%%%%%%%%%%%%%%%%%%%%%%%
\IEEEPARstart{W}{e} consider linear inverse problems that occur in a variety of biomedical imaging applications~\cite{Vonesch.etal2006, Bertero.Boccacci1998,Ribes.Schmitt2008}. In this class of problems, the measurements $\y$ are obtained through the forward model
\begin{equation}\label{eq:discreteLinearModel}
\y=\h\sbf+\n\text{,}
\end{equation}
where $\sbf$ represents the true signal/image. The linear operator $\h$ models the physical response of the acquisition/imaging device and $\n$ is some additive noise. A conventional approach for reconstructing $\sbf$ is to formulate the reconstructed signal $\sbf^{\star}$ as the solution of the optimization problem
\begin{equation}\label{eq:generalReconstructionFormula}
\sbf^{\star} = \arg \underset{\sbf}{\min}~\mathcal{D}(\sbf;\y) + \lambda \mathcal{R}(\sbf)\text{,}
\end{equation}
where $\mathcal{D}(\sbf;\y)$ quantifies the distance separating the reconstruction from the observed measurements, $\mathcal{R}(\sbf)$ measures the regularity of the reconstruction, and  $\lambda > 0$ is the regularization parameter.

In the classical quadratic (Tikhonov-type) reconstruction schemes, one utilizes $\ell_2$-norms for measuring  both the data consistency and the reconstruction regularity~\cite{Ribes.Schmitt2008}. In a general setting, this leads to a smooth optimization problem of the form
\begin{equation}\label{eq:tikhonovRegularization}
\sbf^{\star} = \arg \underset{\sbf}{\min} \| \y - \h\sbf \|_2^2 + \lambda \|\mathbf{R}\sbf\|_2^2\text{,}
\end{equation}
where $\mathbf{R}$ is a linear operator and the formal solution is given by
\begin{equation}\label{eq:solutionTikhonov}
\sbf^{\star} = \left( \h^{\mathrm{T}}\h + \lambda \mathbf{R}^\mathrm{T}\mathbf{R} \right)^{-1}\h^\mathrm{T}\y\text{.}
\end{equation} 
The linear reconstruction framework expressed in~\eqref{eq:tikhonovRegularization}-\eqref{eq:solutionTikhonov} can also be derived from a statistical perspective. Under the hypothesis that $\sbf$ follows a multivariate zero-mean Gaussian distribution with covariance matrix $\mathbf{C}_{\sbf\sbf}=\mathbb{E}\lbrace\sbf\sbf^\mathrm{T}\rbrace$, the operator $\mathbf{C}_{\sbf\sbf}^{-1/2}$ \textit{whitens} $\sbf$ (i.e., renders its components independent). Moreover, if $\n$ is additive white Gaussian noise (AWGN) of variance $\sigma^2$,  the maximum a posteriori (MAP) formulation of the reconstruction problem yields
\begin{equation}\label{eq:gaussianMAP}
\sbf_{\rm MAP} = (\h^{\mathrm{T}}\h+\sigma^2\mathbf{C}_{\sbf\sbf}^{-1})^{-1}\h^\mathrm{T}\y\text{,} 
\end{equation}
which is equal to~\eqref{eq:solutionTikhonov} when $\mathbf{C}_{\sbf\sbf}^{-1/2}=\mathbf{R}$ and $\sigma^2=\lambda$. In the Gaussian scenario, the MAP estimator is known to yield the minimum mean square error (MMSE) solution. The equivalent Wiener solution~\eqref{eq:gaussianMAP} is also applicable for non-Gaussian models with known covariance $\mathbf{C}_{\sbf\sbf}$ and is commonly referred to as the linear minimum mean square error (LMMSE)~\cite{Kay.1993}.

In recent years, the paradigm in variational formulations for signal reconstruction has shifted from the classical linear schemes to the \textit{sparsity-promoting} methods motivated by the observation that many signals that occur naturally have sparse or nearly-sparse representations in some transform domain~\cite{Mallat.2008}. The promotion of sparsity is achieved by specifying well-chosen non-quadratic regularization functionals and results in nonlinear reconstruction. One common choice for the regularization functional is $\mathcal{R}(\mathbf{v})=\|\mathbf{v}\|_1$, where $\mathbf{v}=\mathbf{W}^{-1}\sbf$ represents the coefficients of a wavelet (or a wavelet-like multiscale) transform~\cite{Lustig.etal2007}. An alternative choice is $\mathcal{R}(\sbf)=\|\Lbf\sbf\|_1$, where $\Lbf$ is the discrete version of the gradient or Laplacian operator, with the gradient one being known as total-variation (TV) regularization~\cite{Rudin.etal1992}. Although using the $\ell_1$ norm as regularization functional has been around for some time (for instance, see~\cite{Claerbout.Muir1973,Taylor.etal1979}), it is currently at the heart of sparse signal reconstruction problems. Consequently, a significant amount of research is dedicated to the design of efficient algorithms for nonlinear reconstruction methods~\cite{Zibulevsky.Elad2010}.

The current formulations of sparsity-promoting regularization are based on solid variational principles and are predominantly deterministic. They can also be interpreted in statistical terms as MAP estimators by considering generalized Gaussian or Laplace priors\cite{Bouman.Sauer1993,Choi.Baraniuk1999,Babacan.etal2010}.  These models, however, are tightly linked to the choice of a given sparsifying transform, with the downside that they do not provide further insights on the true nature of the signal.

\subsection{Contributions}

In this paper, we revisit the signal reconstruction problem by specifying upfront a \textit{continuous-domain model} for the signal that is independent from the subsequent reconstruction task/algorithm and apply a \textit{proper discretization scheme} to derive the corresponding MAP estimator. Our approach builds upon the theory of continuous-domain sparse stochastic processes~\cite{Unser.etal2011a}.  In this framework, the stochastic process is defined through an innovation model that can be driven by a \textit{non-Gaussian} excitation~\footnote{It is noteworthy that the theory includes the stationary Gaussian processes.}. The primary advantage of our continuous-domain formulation is that it lends itself to an analytical treatment. In particular, it allows for the derivation of the probability density function (pdf) of the signal in any transform domain, which is typically much more difficult in a purely discrete framework. Remarkably, the underlying class of models also provides us with a strict derivation of the class of admissible regularization functionals which happen to be confined to two categories: Gaussian or sparse. 

The main contributions of the present work are as follows:
\begin{itemize}
\item The introduction of  continuous-domain stochastic models in the formulation of inverse problems. This leads to the use of non-quadratic reconstruction schemes.
\item A general scheme for the proper discretization of these problems. This scheme specifies feasible statistical estimators.
\item The characterization of the complete class of admissible potential functions (prior log-likelihoods) and the derivation of the corresponding MAP estimators. The connections between these estimators and the existing deterministic methods such as TV and $\ell_1$ regularizations are also explained. 
\item A general reconstruction algorithm, based on variable-splitting techniques, that handles different estimators, including the
    nonconvex ones. The algorithm is applied to deconvolution and to the reconstruction of MR and X-ray images.
\end{itemize}

\subsection{Outline}
The paper is organized as follows: In Section~\ref{sec:MeasurementModel}, we explain the acquisition model and obtain the corresponding representation of the signal $\sbf$ and the system matrix $\h$. In Section~\ref{sec:SparseStochasticModels}, we introduce the continuous-domain innovation model that defines a generalized stochastic process. We then statistically specify the discrete-domain counterpart of the innovation model and characterize admissible prior distributions. Based on this characterization,  we derive the MAP estimation as an optimization problem  in Section~\ref{sec:MAP}. In Section~\ref{sec:algorithm}, we provide an efficient algorithm to solve the optimization problem for a variety of admissible priors. Finally, in Section~\ref{sec:numericalResults}, we illustrate our discretization procedure by applying it to a series of deconvolution and of  MR and X-ray image-reconstruction problems. This allows us to compare the effect of different sparsity priors on the solution. 

\subsection{Notations}
Throughout the paper, we assume that the measurement noise is AWGN of variance $\sigma^2$. The input argument $\bx\in\mathbb{R}^d$ of the continuous-domain signals is written inside parenthesis (e.g., $s(\bx)$) whereas, for discrete-domain signals, we employ $\bk\in\mathbb{Z}^d$ and use brackets (e.g., $s[\bk]$). The scalar product is represented by $\langle \cdot, \cdot\rangle$ and $\delta(\cdot)$ denotes the Dirac impulse. 
%%%%%%%%%%%%%%%%%%%%%%%%%%%%%%%%%%%%%%%%%%%%%%%%%%%%%%%%%%%%%%%%%

%%%%%%%%%%%%%%%%%%%%%%%%%%%%%%%%%%%%%%%%%%%%%%%%%%%%%%%%%%%%%%%%%
\section{Measurement Model}\label{sec:MeasurementModel}
%%%%%%%%%%%%%%%%%%%%%%%%%%%%%%%%%%%%%%%%%%%%%%%%%%%%%%%%%%%%%%%%%

%%%%%%%%%%%%%%%%%%%%%%%%%%%% FIGURE SIGNAL MODEL %%%%%%%%%%%%%%%%
\begin{figure}[tb]	
	\begin{minipage}[b]{1.0\linewidth}
  		\centering
 			 \centerline{\includegraphics[width=8.5cm]{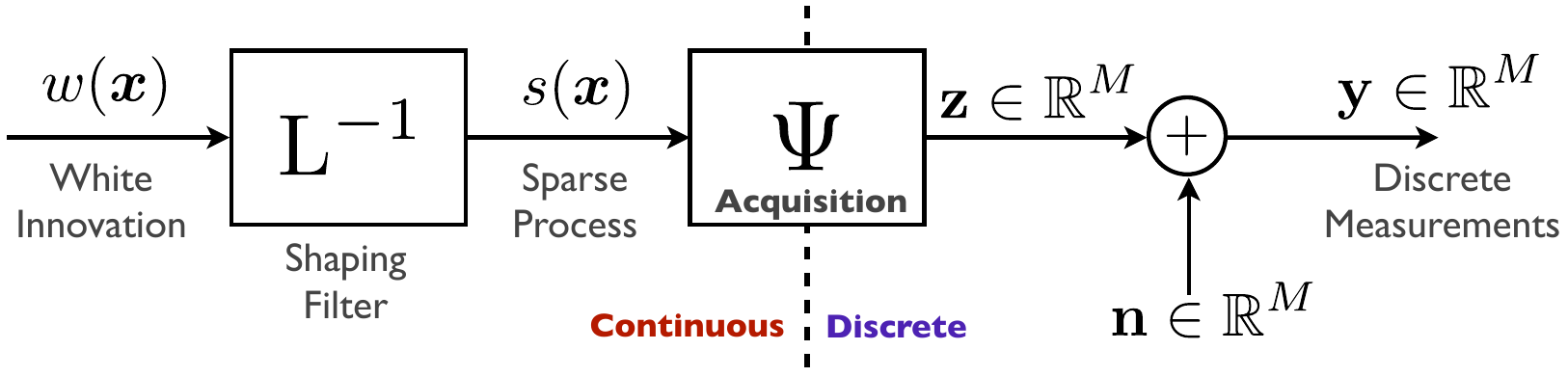}}
	\end{minipage}
	\caption{General form of the linear, continuous-domain measurement model considered in this paper. The signal $s(\bx)$ is acquired through linear measurements of the form $z_m = \left[\mathrm{\Psi}s\right]_m =\langle s, \psi_m \rangle$. The resulting vector $\zbf \in \R^M$ is corrupted with AWGN. Our goal is to estimate the original signal $s$ from noisy measurements $\ybf$ by exploiting the knowledge that $s$ is a realization of a sparse stochastic process that satisfies the innovation model $\Lop s = w$, where $w$ is a non-Gaussian white innovation process.}
	\label{fig:SignalModel}
\end{figure}
%%%%%%%%%%%%%%%%%%%%%%%%%%%%%%%%%%%%%%%%%%%%%%%%%%%%%%%%%%%%%%%

In this section, we develop a discretization scheme that allows us to obtain a tractable representation of continuously-defined measurement problem, with minimal loss of information. Such discrete representation is crucial since the resulting reconstruction algorithms are implemented numerically.

\subsection{Discretization of the Signal}\label{subsec:Discretization}

To obtain a clean analytical discretization of the problem, we consider the
generalized sampling approach using ``shift-invariant''  reconstruction spaces~\cite{Unser.2000}. The advantage of such a representation is that it offers the same type of error control as finite-element methods (i.e., one can make the discretization error arbitrarily small by making the reconstruction grid sufficiently fine).

The idea is to represent the signal $s$ by projecting it onto a reconstruction space. We define our reconstruction space at resolution $T$ as
\begin{small}
\begin{equation}
V_T(\varphi_{\rm int})=\left\{ s_T(\bx) = \sum_{\bk \in \Z^{d}} s\left[\bk\right] \varphi_{\rm int}\left(\frac{\bx}{T} - \bk\right) : s[\bk] \in \ell_\infty(\Z^d)\right\}\hspace{-0.25em},
\end{equation}
\end{small}

\noindent where  $s[\bk]=s(\bx)|_{\bx=T\bk}$, and $\varphi_{\rm int}$ is an interpolating basis function positioned on the reconstruction grid $T\Z^d$.  The interpolation property is $\varphi_{\rm int}(\bk)=\delta[\bk]$. For the representation of $s$ in terms of its samples $s[\bk]$ to be stable and unambiguous, $\varphi_{\rm int}$ has to be a valid Riesz basis for $V_T(\varphi_{\rm int})$. Moreover, to guarantee that the approximation error decays as a function of $T$, the basis function should satisfy the partition of unity property~\cite{Unser.2000}
\begin{equation}
\sum_{\bk \in \Z^d} \varphi_{\rm int}(\bx - \bk) = 1,~\forall \bx \in \R^d\text{.}
\end{equation}
The projection of the signal onto the reconstruction space $V_T(\varphi_{\rm int})$ is given by
\begin{equation}
P_{V_T}s(\bx) = \sum_{\bk \in \Z^d} s(T\bk)\varphi_{\rm int}\left(\frac{\bx}{T} - \bk\right)\hspace{-0.25em},
\end{equation}
with the property that $P_{V_T}P_{V_T}s(\bx) = P_{V_T}s(\bx)$ (since $P_{V_T}$ is a projection operator). To simplify the notation, we shall use a unit sampling $T = 1$ with the implicit assumption that the sampling error is negligible. (If the sampling error is large, one can use a finer sampling and rescale the reconstruction grid appropriately.) Thus, the resulting discretization is
\begin{equation}
\label{equ:Expansion}
s_1(\bx) = P_{V_1}s(\bx) = \sum_{\bk \in \Z^d} s[\bk] \varphi_{\rm int}(\bx - \bk)\text{.}
\end{equation}
To summarize, $s_1(\bx)$ is the discretized version of the original signal $s(\bx)$ and it is uniquely described by the samples $s[\bk]=s(\bx)|_{\bx=\bk}$ for ${\bk \in \Z^d}$.  The main point is that the reconstructed signal is represented in terms of samples even though the problem is still formulated in the continuous-domain.
\subsection{Discrete Measurement Model}

By using the discretization scheme in~\eqref{equ:Expansion}, we are now ready to formally link the continuous model in Figure~\ref{fig:SignalModel} and the corresponding discrete linear-inverse problem. Although the signal representation~\eqref{equ:Expansion} is an infinite sum, in practice we restrict ourselves to a subset of $N$ basis functions with $\bk \in \Omega$,  where $\Omega$ is a discrete set of integer coordinates in a region of  interest (ROI).  Hence, we rewrite~\eqref{equ:Expansion} as
\begin{equation}
s_1(\bx) =  \sum_{\bk \in \Omega} s[\bk] \varphi_{\bk}(\bx),
\end{equation}
where $\varphi_{\bk}(\bx)$ corresponds to $\varphi_{\rm int}(\bx-\bk)$ up to modifications at the boundaries (periodization or Neumann boundary condition). 

We first consider a noise-free signal acquisition. The general form of a linear, continuous-domain noise-free measurement system is
\begin{equation}
\label{equ:ContinuousMeasurement}
z_m = \int_{\R^d} s(\bx)\psi_m(\bx)\mathrm{d}\bx, \,\, (m = 1, \dots, M)
\end{equation}
where $s(\bx)$ is the original signal, and the measurement function $\psi_m(\bx)$ represents the spatial response of the $m$th detector which is application dependent as we shall explain in Section~\ref{sec:numericalResults}.

By substituting the signal representation~\eqref{equ:Expansion} into~\eqref{equ:ContinuousMeasurement}, we discretize the measurement model and write it in matrix-vector form as
\begin{equation}
\ybf = \zbf + \n = \h\sbf + \n,
\end{equation}
where $\ybf$ is the $M$-dimensional measurement vector, $\sbf = \left(s[\bk]\right)_{\bk \in \Omega}$ is the $N$-dimensional signal vector, $\n$ is the $M$-dimensional noise vector, and $\h$ is the $M \times N$ system matrix whose entry $(m,\bk)$ is given by
\begin{equation}
\label{equ:matrixVector}
\left[\h\right]_{m, \bk} = \langle \psi_m, \varphi_\bk\rangle = \int_{\R^d} \psi_m(\bx)\varphi_\bk(\bx)\mathrm{d}\bx.
\end{equation}
This allows us to specify the discrete linear forward model given in \eqref{eq:discreteLinearModel} which is compatible with the continuous-domain formulation. The solution of this problem yields the representation $s_1(\bx)$ of $s(\bx)$ which is parameterized in terms of the signal samples $\sbf$. Having the forward model explained, our next aim is to obtain the statistical distribution of $\sbf$.
%%%%%%%%%%%%%%%%%%%%%%%%%%%%%%%%%%%%%%%%%%%%%%%%%%%%%%%%%%%%%%%%%

%%%%%%%%%%%%%%%%%%%%%%%%%%%%%%%%%%%%%%%%%%%%%%%%%%%%%%%%%%%%%%%%%
\section{Sparse Stochastic Models}\label{sec:SparseStochasticModels}
%%%%%%%%%%%%%%%%%%%%%%%%%%%%%%%%%%%%%%%%%%%%%%%%%%%%%%%%%%%%%%%%%
We now proceed by introducing our stochastic framework which will provide us with a signal prior. For that purpose,  we assume that $s(\bx)$ is a realization of a stochastic process that is defined as the solution of a linear stochastic differential equation (SDE) with a  driving term that is not necessarily Gaussian. Starting from such a continuous-domain model, we aim at obtaining the statistical distribution of the sampled version of the process (discrete signal) that will be needed to formulate estimators for the reconstruction problem.

\subsection{Continuous-Domain Innovation Model}
As mentioned in Section~\ref{sec:Introduction}, we specify our relevant class of signals as the solution of an SDE in which the process $s$ is assumed to be whitened by a linear operator. This model takes the form
\begin{equation}\label{continuousInnovation}
    \Lop s=w\text{,}
\end{equation}
where  $w$ is a continuous-domain white innovation process (the driving term), and ${\rm L}$ is a (multidimensional) differential operator. The right-hand side of (\ref{continuousInnovation}) represents the unpredictable part of the process, while $\Lop$ is called the whitening operator. Such models are standard in the classical theory of stationary Gaussian processes~\cite{Papoulis.1991}. The twist here is that the driving term $w$ is not necessarily Gaussian. Moreover, the underlying differential system is potentially unstable to allow for self-similar models. 

In the present model, the process $s$ is characterized 
by the formal solution $s={\rm L}^{-1}w$, where ${\rm L}^{-1}$ is an appropriate right inverse of $\Lop$. The operator ${\rm L}^{-1}$ amounts to some generalized ``integration" of the innovation $w$. The implication is that the correlation structure of the stochastic process $s$ is determined by the shaping operator $\Lop^{-1}$, whereas its statistical properties and sparsity structure is determined by the driving term $w$.   As an example in the one-dimensional setting, the operator $\Lop$ can be chosen as the first-order continuous-domain derivative operator  $\Lop=\Dop$. For multidimensional signals, an attractive class of operators is the fractional Laplacian $(-\Delta)^{\frac{\gamma}{2}}$ which is invariant to translation, dilation, and rotation in $\R^d$~\cite{Sun.Unser2012}. This operator gives rise to ``$1/ \| \boldsymbol{\omega} \|^\gamma$''-type power spectrum and is frequently used to model certain types of images~\cite{Mandelbrot.1983,Huang.Mumford1999}.

The mathematical difficulty is that the innovation $w$ cannot be interpreted as an ordinary function because it is highly singular. The proper framework for handling such singular objects is Gelfand and Vilenkin's theory of generalized stochastic processes~\cite{Gelfand.Vilenkin1964}.  In this framework, the stochastic process $s$ is observed by means of scalar-products $\langle s,\varphi \rangle$ with $\varphi \in \mathcal{S}(\R^d)$, where $\mathcal{S}(\R^d)$ denotes the Schwartz class of rapidly decreasing test functions. Intuitively, this is analogous to measuring an intensity value at a pixel after integration through a CCD detector.

A fundamental aspect of the theory is that the driving term $w$ of the innovation model~\eqref{continuousInnovation} is uniquely specified in terms of its L\'evy exponent $f(\cdot)$. 
%At this point, we provide the following definitions related to L\'evy exponents, which play a central role in the specification of sparse stochastic processes. 
%\begin{defn}\label{def:levyfunc}
%A complex-valued function $f: \R \rightarrow \C$ is said to be conditionally positive-definite of order one if
%$$
%\sum_{m=1}^N \sum_{n=1}^N f(\omega_m-\omega_n) \xi_m\overline{\xi}_n \ge 0
%$$
%under the condition $\sum_{m=1}^N \xi_m=0$ for every possible choice of $\omega_1,\dots,\omega_N \inR$, $\xi_1,\dots,\xi_N \in \C%$ and $N \in \N$. 
%\end{defn}
%------------------------------------------------
\begin{defn}\label{def:levyfunc}
A complex-valued function $f: \R \rightarrow \C$ is a valid L\'evy exponent iff. it satisfies the three following conditions:
\begin{enumerate}
\item it is continuous;
\item it vanishes at the origin;
\item it is conditionally positive-definite of order one in the sense that
$$
\sum_{m=1}^N \sum_{n=1}^N f(\omega_m-\omega_n) \xi_m\overline{\xi}_n \ge 0
$$
under the condition $\sum_{m=1}^N \xi_m=0$ for every possible choice of $\omega_1,\dots,\omega_N \inR$, $\xi_1,\dots,\xi_N \in \C$, and $N \in \N\setminus\{0\}$. 
\end{enumerate}
\end{defn}
An important subset of L\'evy exponents are the $p$-admissible ones, which are central to our formulation.
\begin{defn}
A L\'evy exponent $f$ with derivative $f'$ is called $p$-admissible if it satisfies the inequality
\begin{equation*}
|f(\omega)|+|\omega| |f'(\omega)| \leq C |\omega|^p
\end{equation*}
for some constant $C>0$ and $0<p\leq2$.
\end{defn}
A typical example of a $p$-admissible L\'evy exponent is $f(\omega)=-s_0|\omega|^\alpha$ with $s_0 > 0$. The simplest case is $f_{\rm Gauss}(\omega)=-\half|\omega|^2$; it will be used to specify Gaussian processes.

Gelfand and Vilenkin have characterized the whole class of continuous-domain white innovation and have shown that they are fully specified by the generic characteristic form 
\begin{align}
    \CH_w(\varphi) & = \mathbb{E}\left\lbrace {\rm e}^{{\rm j}\langle w,\varphi \rangle}\right\rbrace \nonumber \\
    			        & = \exp\left( \int_{\R^d} f(\varphi(\bx)){\rm d}\bx\right)\hspace{-0.25em }\text{,} \label{eq:CharForm}
\end{align}
where $f$ is the corresponding L\'evy exponent of the innovation process $w$. The powerful aspect of this characterization is that $\CH_w$ is indexed by a test function $\varphi\in\mathcal{S}$ rather than by a scalar (or vector) Fourier variable $\omega$. As such, it constitutes the {\em infinite-dimensional} generalization of the characteristic function of a conventional random variable. 

Recently, Unser et al. characterized the class of stochastic processes that are solutions of (\ref{continuousInnovation}) where $\Lop$ is a linear shift-invariant (LSI) operator and $w$ is a member of the class of so-called L\'evy noises~\cite[Theorem 3]{Unser.etal2011a}. 
\begin{theo}\label{theo:Unser.etal}
%Let $s$ be the solution of $\Lop s=w$ where $w$ is a L\'evy noise as specified by~\eqref{eq:CharForm}. 
Let $w$ be a L\'evy noise as specified by~\eqref{eq:CharForm} and $\Lop^{-1*}$ be a left inverse of the adjoint operator $\Lop^\ast$ such that
either one of the conditions below is met:
\begin{enumerate}
\item $\Lop^{-1*}$ is a continuous linear map from $\mathcal{S}(\R^d)$ into itself;
\item $f$ is $p$-admissible and $\Lop^{-1*}$ is a continuous linear map from $\mathcal{S}(\R^d)$ into $L_p(\R^d)$; that is,
\begin{equation*}
\|\Lop^{-1*}\varphi\|_{L_p}<C \|\varphi\|_{L_p}\text{,} \quad \forall \varphi \in \mathcal{S}(\R^d) 
\end{equation*}
for some constant $C$ and some $p\geq1$.
\end{enumerate} 
Then, $s=\Lop^{-1}w$ is a well-defined generalized stochastic process over the space of tempered distributions $\mathcal{S}^{\prime}(\R^d)$ and is uniquely characterized by its characteristic form
\begin{equation}\label{eq:charform}
\CH _s(\varphi)=\mathbb{E}\left\lbrace {\rm e}^{{\rm j}\langle s,\varphi \rangle}\right\rbrace=\exp\left( \int_{\mathbb{R}^d} f\big(\Lop^{-1*}\varphi(\bx)\big){\rm d}\bx \right)\hspace{-0.25 em}\text{.}
\end{equation}
It is a (weak) solution of the stochastic differential equation $\Lop s=w$ in the sense that $\langle \Lop s,\varphi \rangle=\langle w,\varphi \rangle$ for all $\varphi \in \mathcal{S}(\R^d)$.
 \end{theo}

Before we move on, it is important to emphasize that L\'evy exponents are in one-to-one correspondence with the so-called infinitely divisible (i.d.) distributions~\cite{Sato.1994}. 
\begin{defn}
A generic pdf $p_X$ is infinitely divisible if, for any positive integer $n$, it can be represented as the $n$-fold convolution $(p\ast \dots \ast p)$ where $p$ is a valid pdf.
\end{defn}
\begin{theo}[L\'evy-Schoenberg] 
\label{Th:levySchoenberg}
Let $\hat p_{X}(\omega)=\mathbb{E}\{{\rm e}^{{\rm j} \omega X}\}=\int_\R {\rm e}^{{\rm j} \omega x} p_{X}(x)\dint x$ be the characteristic function of an
infinitely divisible  random variable $X$. Then,
$$f(\omega)=\log \hat p_X(\omega)$$
is a L\'evy exponent in the sense of Definition \ref{def:levyfunc}. Conversely, if $f(\omega)$ is a valid L\'evy exponent, then the inverse Fourier integral $$
p_X(x)=\int_\R {\rm e}^{f(\omega)} {\rm e}^{-{\rm j}\omega x} \frac{ \dint \omega}{2 \pi}
$$
yields the pdf of an i.d. random variable.
\end{theo}

Another important theoretical result is that it is possible to specify the complete family of i.d. distributions thanks to the celebrated L\'evy-Khintchine representation~\cite{Steutel.Harn2004} which provides a constructive method for defining L\'evy exponents. This tight connection will be essential for our formulation and limits us to a certain family of prior distributions.

\subsection{Statistical Distribution of Discrete Signal Model}
The interest is now to statistically characterize the discretized signal described in Section~\ref{subsec:Discretization}. To that end, the first step is to formulate a discrete version of the continuous-domain innovation model~\eqref{continuousInnovation}. Since, in practical applications, we are only given the samples $s[\bk]_{\bk \in \Omega}$ of the signal, we obtain the discrete-domain innovation model by applying to them the discrete counterpart  ${\rm L}_{\rm d}$ of the whitening operator ${\rm L}$. The fundamental requirement for our formulation is that the composition of $\Lop_{\rm d}$ and $\Lop^{-1}$ results in a stable, shift-invariant operator whose impulse response is well localized~\cite{Unser.etal2011b}
\begin{equation}\label{eq:Bspline}
\left( {\rm L}_{\rm d} {\rm L}^{-1}\delta \right) (\bx)=\beta_{\rm L}(\bx) \in L_1(\R^d)\text{.}
\end{equation}
The function $\beta_\Lop$ is the generalized B-spline associated with the operator $\Lop$. Ideally, we would like it to be maximally localized. 

To give more insight, let us consider $\Lop=\Dop$ and $\Lop_{\rm d}=\Dop_{\rm d}$ (the finite-difference operator associated to \Dop). Then, the associated B-spline is $\beta_\Dop(x)=\Dop_{\rm d}\mathds{1}_+(x)=\mathds{1}_+(x)-\mathds{1}_+(x-1)$ where $\mathds{1}_+(x)$ is the  unit step (Heaviside) function. Hence, $\beta_{\Dop}(x)={\rm rect}(x-\half)$ is a causal rectangle function (polynomial B-spline of degree $0$). 

The practical consequence of~\eqref{eq:Bspline} is
\begin{align}\label{eq:discreteInnovations}
u={\rm L}_{\rm d}s = {\rm L}_{\rm d}\Lop^{-1}w= \beta_\Lop \ast w\text{.}		   
\end{align}
Since $(\beta_\Lop \ast w)(\bx)=\langle w,\beta_\Lop^\vee(\cdot-\bx)  \rangle$ where $\beta_\Lop^\vee(\bx)=\beta_\Lop(-\bx)$ is the space-reversed version of $\beta_\Lop$, it can be inferred from~\eqref{eq:discreteInnovations} that the evaluation of the samples of ${\rm L}_{\rm d}s$ is equivalent to the observation of the innovation through a B-spline window. 

From a system-theoretic point of view, $\Lop_{\rm d}$ is understood as a finite impulse response (FIR) filter. This impulse response is of the form $\sum_{\bk \in \Omega}d[\bk]\delta(\cdot-\bk)$ with some appropriate weights $d[\bk]$. Therefore, we write the discrete counterpart of the continuous-domain innovation variable as
\begin{equation*}
u[\bk] = {\rm L}_{\rm d}s(\bx)|_{\bx=\bk}=\sum_{\bk^{\prime}\in\Omega}d[\bk^{\prime}]s(\bk-\bk^{\prime})\text{.}
\end{equation*}
This allows us to write  in matrix-vector notation the discrete-domain version of the innovation model (\ref{continuousInnovation}) as
\begin{equation} 
\ubf=\Lbf\sbf\text{,}
\end{equation}
where $\sbf=\left(s[\bk]\right)_{\bk \in \Omega}$ represents the discretization of the stochastic model with $s[\bk]=s(\bx)|_{\bx=\bk}$ for ${\bk \in \Omega}$, $\Lbf:\mathbb{R}^N \rightarrow \mathbb{R}^N$ is the matrix representation of $\Lop_{\rm d}$, and $\ubf=\left( u[\bk]\right)_{\bk \in \Omega}$ is the discrete innovation vector.

We shall now rely on~\eqref{eq:charform} to derive the pdf of the discrete innovation variable, which is one of the key results of this paper.
\begin{theo}\label{theo:infiniteDivisible}
Let $s$ be a stochastic process whose characteristic form is given by~\eqref{eq:charform} where $f$ is a $p$-admissible L\'evy exponent, and $\beta_\Lop=\Lop_{\rm d}\Lop^{-1}\delta\in L_p(\R^d)$ for some $p\in(0,2]$.
Then, $u=\Lop_{\rm d}s$ is stationary and infinitely divisible. Its first-order pdf is given by
\begin{equation}\label{eq:theoremResult}
p_U(u)=\int_\R \exp\left(f_{\beta_\Lop^\vee}(\omega)\right) {\rm e}^{{\rm j} \omega u} \frac{\dint \omega}{2 \pi}\text{,}
\end{equation}
with L\'evy exponent
\begin{equation}\label{eq:levyExponentSpline}
f_{\beta_\Lop^\vee}(\omega)= \log \hat p_U(\omega)=\int_{\R^d} f\big(\omega\beta_\Lop^\vee(\bx)\big) \dint \bx\text{,}
\end{equation}
which is $p$-admissible as well. 
\end{theo}
%----------------
\begin{proof}
Taking \eqref{eq:discreteInnovations} into account, we derive the characteristic form of $u$ which is given by
\begin{align}
\CH_u(\varphi)&=\mathbb{E}\{{\rm e}^{{\rm j} \langle u, {\varphi} \rangle} \}=\mathbb{E}\{{\rm e}^{{\rm j} \langle \beta_\Lop \ast w, {\varphi} \rangle} \}=\mathbb{E}\{{\rm e}^{{\rm j} \langle w,\beta_\Lop^\vee \ast {\varphi} \rangle} \} \nonumber\\
&=\CH_w(\beta_\Lop^\vee\ast\varphi) \nonumber\\
&=\exp\left( \int_{\R^d} f\big(\beta_\Lop^\vee  \ast\varphi(\bx)\big) \dint \bx\right).
\label{eq:charu}
\end{align}
The fact that $u$ is stationary is equivalent to $\CH_u(\varphi)=\CH_u\big(\varphi(\cdot-\bx_0)\big)$ for any $\bx_0\inR^d$, which is established by a simple change of variable in (\ref{eq:charu}).
We now consider the random variable $U=\langle u, \delta \rangle=\langle w, \beta_\Lop^\vee\rangle$. Its characteristic function is obtained as
\begin{align*}
\hat p_U(\omega)&
=\mathbb{E}\{{\rm e}^{{\rm j}  \omega U} \} =\mathbb{E}\{{\rm e}^{{\rm j} \langle w, \omega \beta_\Lop^\vee\rangle} \}\\
&=\CH_w(\omega\beta_\Lop^\vee) \\
&=\exp\left( f_{\beta_\Lop^\vee}(\omega)\right) 
\label{eq:padmod}
\end{align*}
where the substitution $\varphi=\omega\beta_\Lop^\vee$ in $\CH_w(\varphi)$ is valid since $\CH_w$ is a continuous
functional on $L_p(\R^d)$ as a consequence of the $p$-admissibility condition. To prove that $f_{\beta_\Lop^\vee}(\omega)$ is a $p$-admissible L\'evy exponent, we start by establishing the bound
\begin{align}
C \|\varphi\|_{L_p}^p |\omega|^p \ge&  \int_{\R^d} \left|f\big(\omega\beta_\Lop^\vee(\bx)\big)\right| \dint \bx \nonumber\\
&\, +|\omega| \int_{\R^d} \left| f'\big(\omega\beta_\Lop^\vee(\bx)\big)\varphi(\bx) \right|\dint \bx \nonumber\\
\ge& \left|f_{\beta_\Lop^\vee}(\omega)\right| + |\omega| \left|f'_{\beta_\Lop^\vee}(\omega)\right|,
\end{align}
which follows from the $p$-admissibility of $f$. We are also relying on Lebesgue's dominated convergence theorem to move the derivative with respect to $\omega$ inside the integral that defines $f_{\beta_\Lop^\vee}(\omega)$. In particular, (\ref{eq:padmod}) implies that $f_{\beta_\Lop^\vee}$ is continuous and vanishes at the origin. The last step is to establish its conditional positive definiteness which is achieved by interchanging the order of summation. We write
\begin{align}
&\sum_{m=1}^N \sum_{n=1}^N f_{\beta_\Lop^\vee}(\omega_m-\omega_n) \xi_m\overline{\xi}_n =  \\
&\int_{\R^d} \underbrace{\sum_{m=1}^N \sum_{n=1}^N f\big(\omega_m\beta_\Lop^\vee(\bx)-\omega_n\beta_\Lop^\vee(\bx)\big) \xi_m\overline{\xi}_n}_{\ge 0} \dint \bx \ge 0 \nonumber 
\end{align}
under the condition $\sum_{m=1}^N \xi_m=0$ for every possible choice of $\omega_1,\dots,\omega_N \inR$, $\xi_1,\dots,\xi_N \in \C,$ and $N \in \N\setminus\{0\}$.
\end{proof}

The direct consequence of Theorem~\ref{theo:infiniteDivisible} is that the primary statistical features of $\ubf$ is directly related to the continuous-domain innovation process $w$ via the L\'evy exponent. This implies that the sparsity structure (tail behavior of the pdf and/or presence of a mass distribution at the origin) is primarily dependent upon $f$. The important conceptual aspect, which follows from the L\'evy-Schoenberg theorem, is that the class of admissible pdfs is restricted to the family of i.d. laws since $f_{\beta_\Lop^\vee}(\omega)$, as given by~\eqref{eq:levyExponentSpline}, is a valid L\'evy exponent. We emphasize that this result is attained by taking advantage of considerations in the continuous-domain. 

\subsection{Specific Examples}
We now would like to illustrate our formalism by presenting some examples. If we choose $\Lop={\rm D}$, then the solution of~\eqref{continuousInnovation} with the boundary condition $s(0)=0$ is given by  
$$
s(x)=\int_0^x w(x^{\prime}){\rm d}x^{\prime}
$$
and is a L\'evy process. It is noteworthy that the L\'evy processes---a fundamental and well-studied family of stochastic processes---include Brownian motion and Poisson processes which are commonly used to model random physical phenomena~\cite{Sato.1994}. In this case, $\beta_{\Dop}(x)={\rm rect}(x-\half)$ and the discrete innovation vector $\ubf$ is obtained by
$$
u[k] = \langle w, {\rm rect}(\cdot+\half-k) \rangle,
$$
representing the so-called ``stationary independent increments''. Evaluating~\eqref{eq:levyExponentSpline} together with $f(0)=0$ (see Definition~\ref{def:levyfunc}), we  obtain
$$
f_{\beta_\Dop^\vee}(\omega)=\int_{-1}^0 f(\omega){\rm d}x=f(\omega).
$$
In particular, we generate L\'evy processes with Laplace-distributed increments by choosing $f(\omega)={\rm log}(\frac{\tau^2}{\tau^2+\omega^2})$ with the scale parameter $\tau>0$. To see that, we write ${\rm exp}(f_{\beta_\Dop^\vee}(\omega)) = \hat p_U(\omega) = \frac{\tau^2}{\tau^2+\omega^2}$ via~\eqref{eq:levyExponentSpline}. The inverse Fourier transform of this rational function is known to be 
$$
p_U(u)=\frac{\tau}{2}{\rm e}^{-\tau|u|}\text{.}
$$
Also, we rely on Theorem~\ref{theo:infiniteDivisible} in a more general aspect. For instance, a special case of  interest is the Gaussian (nonsparse) scenario where $f_{\rm Gauss}(\omega)=-\half|\omega|^2$. Therefore, one gets $f_{\beta_\Lop^\vee}(\omega) = \log \hat p_U(\omega) = -\half\omega^2\|\beta_{\rm L}\|_2^2$ from~\eqref{eq:levyExponentSpline}. Plugging this into~\eqref{eq:theoremResult}, we deduce that the discrete innovation vector is zero-mean Gaussian with variance $\| \beta_{\Lop}\|_2^2$ (i.e., $p_U(u) = \mathcal{N}(0,\| \beta_{\Lop} \|_2^2)$).

Additionally, when $f(\omega)=\frac{-|\omega|^\alpha}{2}$ with $\alpha\in[1,2]$, one finds that $f_{\beta_\Lop^\vee}(\omega) = \log \hat p_U(\omega) =  -  \frac{|\omega|^\alpha}{2} \| \beta_{\rm L} \|_{L_\alpha}^{\alpha}$. This indicates that $\ubf$ is a symmetric $\alpha$-stable ($S\alpha S$) distribution with scale parameter $s_0=\| \beta_{\rm L}\|_{L_\alpha}^\alpha$. For $\alpha=1$, we have the Cauchy distribution (or Student's with $r=1/2$). For other i.d. laws, the inverse Fourier transformation~\eqref{eq:theoremResult} is often harder to compute analytically, but it can still be performed numerically to determine $p_U(u)$ (or its corresponding potential function $\Phi_U=-{\rm log} \, p_U$ ). In general, $p_U$ will be i.d. and will typically imply heavy tails. Note that heavy-tailed distributions are compressible~\cite{Amini.etal2011}.
%%%%%%%%%%%%%%%%%%%%%%%%%%%%%%%%%%%%%%%%%%%%%%%%%%%%%%%%%%%%%%%%%

%%%%%%%%%%%%%%%%%%%%%%%%%%%%%%%%%%%%%%%%%%%%%%%%%%%%%%%%%%%%%%%%%
\section{Bayesian Estimation}\label{sec:MAP}
We now use the results of Section~\ref{sec:SparseStochasticModels} to derive solutions to the reconstruction problem in some well-defined statistical sense. To that end, we concentrate on the MAP solutions that are presently derived under the decoupling assumption that the components of $\ubf$ are independent and identically distributed (i.i.d.).  In order to reconstruct the signal, we seek an estimate of $\sbf$ that maximizes the posterior distribution $p_{S|Y}$ which depends upon the prior distribution $p_S$, assumed to be proportional to $p_U$ (since $\ubf=\Lbf \sbf$). The direct application of Bayes' rule is 
\begin{align*}
p_{S \mid Y}(\sbf \mid \y) &\propto p_N(\y-\h\sbf)p_U({\bf u}) \\
& \propto \exp\left( -\frac{\| \y-\h\sbf \|^2}{2\sigma^2} \right) \prod_{\bk\in\Omega} p_U\big([\Lbf\sbf]_\bk\big)\text{.}
\end{align*}
Then, we write the MAP estimation for $\sbf$ as
\begin{align}
\sbf_{\rm MAP}  
& =  
\arg\underset{\sbf}{\max} \; p_{S \mid Y}(\sbf \mid \y) \notag \\
& =  
 \arg\underset{\sbf}{\min} \; \left( \half\|\h\sbf-\y\|_2^2 +\sigma^2 \sum_{\bk \in \Omega } \Phi_U\big([\Lbf\sbf]_\bk\big)\right)\text{,} \label{eq:MAP}
\end{align}
where $\Phi_U(x)=-{\rm log} \, p_U(x)$ is called the \textit{potential function} corresponding to $p_U$. Note that~\eqref{eq:MAP} is compatible with the standard form of the variational reconstruction formulation given in~\eqref{eq:generalReconstructionFormula}. In the next section, we focus on the potential functions.

%=== BEGIN TABLE ===
\begin{table*}
   \caption{Four members of the family of infinitely divisible distributions and the corresponding potential functions.}
     \begin{center}
     \begin{tabular*}{0.70\textwidth}[ht]{ c  c  c c}
     \hline\hline
                        & $ p_U(x) $                                          & $\Phi_U (x)$    & Property \\ \hline
                        
    Gaussian  &  $ \frac{1}{\sigma_0\sqrt{2\pi}}e^{-x^2/2\sigma_0^2} $ & $M_1x^2+C_1$  & smooth, convex\\ 
    
    Laplace   &  $ \frac{\tau}{2} e^{-\tau|x|}$ & $M_2|x|+C_2$ & nonsmooth, convex\\
    
    Student's &  $ \frac{1}{\epsilon B(r,\frac{1}{2})} \left({\frac{1}{(x/\epsilon)^2+1}}\right)^{r+\frac{1}{2}}$ &  $M_3{\rm log}\left(\frac{x^2+ \epsilon^2}{\epsilon^2}\right)+C_3$  & smooth, nonconvex\\
    
    Cauchy  & $\frac{1}{\pi s_0}\frac{1}{(x/s_0)^2+1}$ &${\rm log}(\frac{x^2+s_0^2}{s_0^2})+C_4$ & smooth, nonconvex \\
    \hline\hline
  \end{tabular*}
  \end{center}
  \label{potentials}  
\end{table*}
%=== END TABLE ===
%=== BEGIN FIGURE ===%
\begin{figure*}[t]
\centering
\subfigure[]{
\raisebox{1.9mm}
{\includegraphics[scale=0.5]{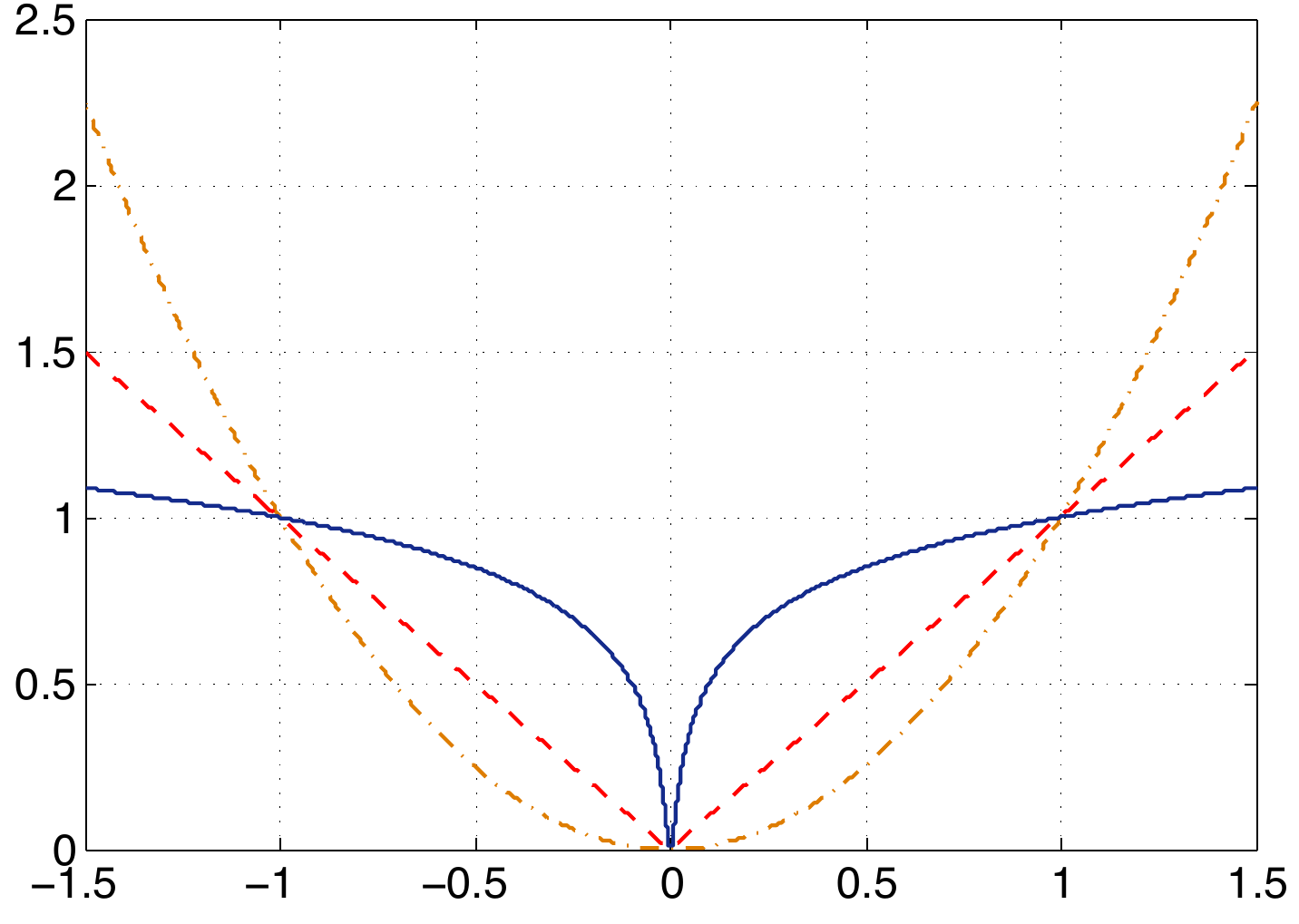}}
\label{fig:potential}
}
\subfigure[ ]{
\includegraphics[scale=0.51]{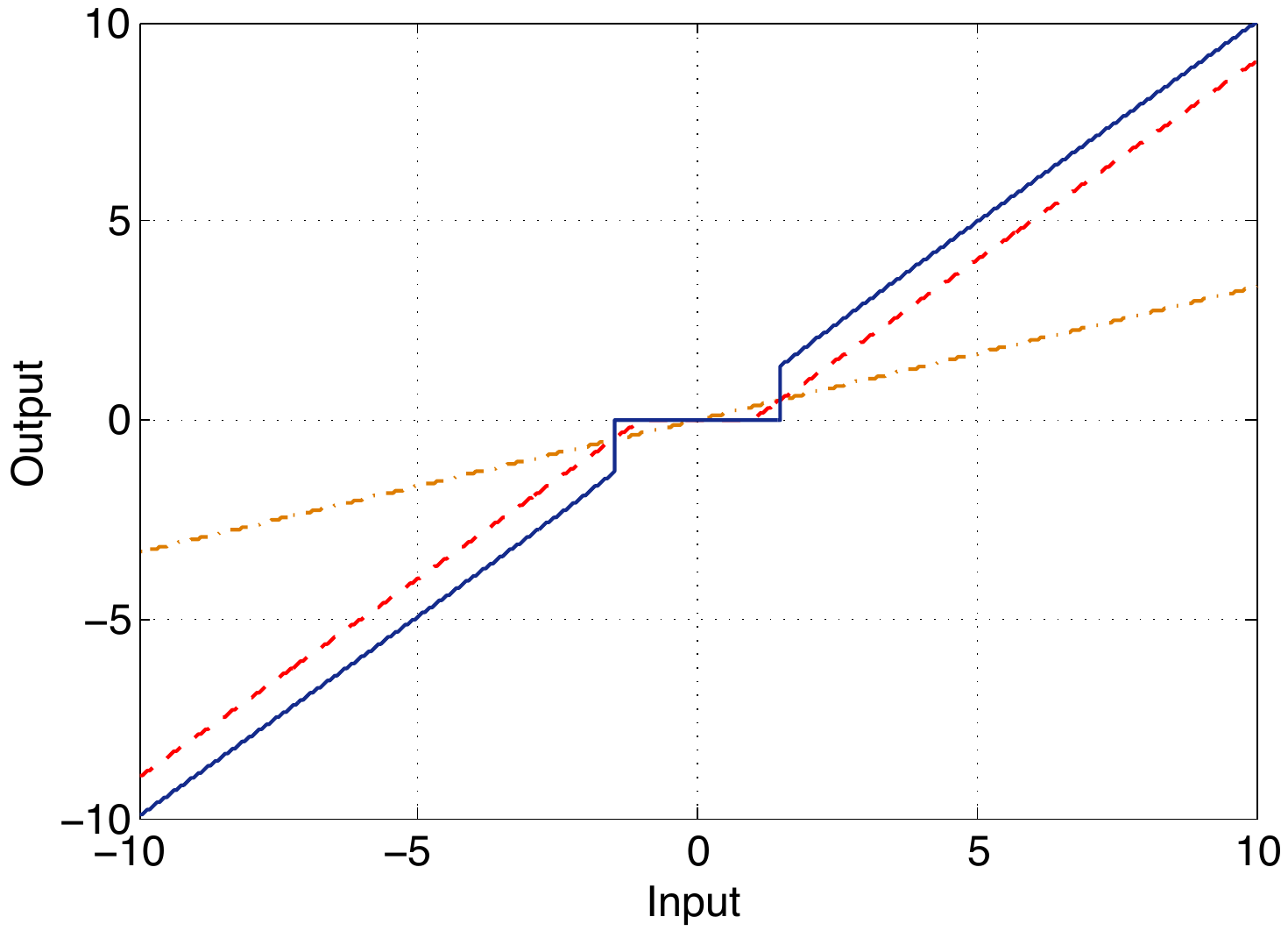}
\label{fig:prox}
}
\caption[Optional caption for list of figures]{Potential functions (a) and the corresponding proximity operators (b) of different estimators: Gaussian estimator (dash-dotted), Laplacian estimator (dashed), and Student's estimator with $\epsilon=10^{-2}$ (solid). The multiplication factors are set such that  $\Phi_U(1)=1$ for all potential functions. }
\label{fig:potentialFunctions}
\end{figure*}
%=== END FIGURE===%

\subsection{Potential Functions}
Recall that, in the current Bayesian formulation, the potential function $\Phi_U(x)=-{\rm log}\,p_U(x)$ is specified by the L\'evy exponent  $f_{\beta_\Lop^\vee}$ which is itself in direct relation with the continuous-domain innovation $w$ via~\eqref{eq:levyExponentSpline}. For illustration purposes, we consider  three members of the i.d. family: Gaussian, Laplace, and Student's (or, equivalently, Cauchy) distributions.  We provide the potential functions for these priors in Table~\ref{potentials}. The exact values of the constants $C_1$, $C_2$, and $C_3$ and the positive scaling factors $M_1$, $M_2$, and $M_3$  have been omitted since they are irrelevant to the optimization problem. On one hand, we already know that the Gaussian prior does not correspond to a sparse reconstruction. On the other hand, the Student's prior has a slower tail decay and promotes sparser solutions than the Laplace prior. Also, to provide a geometrical intuition of how the Student's prior increases the sparsity of the solution, we plot the potential functions for Gaussian, Laplacian, and Student's (with $\epsilon=10^{-2}$) estimators in Figure~\ref{fig:potentialFunctions}. By looking at Figure~\ref{fig:potentialFunctions}, we see that the Student's estimator penalizes small values more than the Laplacian or Gaussian counterparts do. Conversely, it penalizes the large values less. 

Let us point out some connections between the general estimator (\ref{eq:MAP}) and the standard variational methods.  The first  quadratic potential function (Gaussian estimator) yields the classical Tikhonov-type regularizer and produces a stabilized linear solution, as explained in Section~\ref{sec:Introduction}. The second potential function (Laplace estimator) provides the $\ell_1$-type regularizer. Moreover, the well-known TV regularizer~\cite{Rudin.etal1992} is obtained if the operator $\Lop$ is a first-order derivative operator. Interestingly, the third log-based potential (Student's estimator) is linked to the limit case of the $\ell_p$ relaxation scheme as $p \rightarrow 0$~\cite{Wipf.Nagarajan2010}. To see the relation, we note that minimizing $\lim_{p\rightarrow 0}\sum_i |x_i|^p$ is equivalent to minimizing $\lim_{p \rightarrow 0} \sum_i \frac{|x_i|^p - 1}{p}$. As pointed out in~\cite{Wipf.Nagarajan2010},  it holds that
\begin{align}\label{eq:logBound}
\lim_{p \rightarrow 0} \sum_i \frac{|x_i|^p - 1}{p} &= \sum_i{\rm log}|x_i| =  \sum_i \half {\rm log}|x_i|^2 \nonumber\\
& \leq \frac{1}{2} \sum_i {\rm log}(x_i^2 + \kappa)
\end{align}
for any $\kappa \geq 0$. The key observation is that the upper-bounding log-based potential function in~\eqref{eq:logBound} is interpretable as a Student's prior. This kind of regularization has been considered by different authors (see~\cite{Chartrand.Yin2008, Zhang.Kingsbury2010,Kamilov.etal2012} and also~\cite{Candes.etal2008} where the authors consider a similar log-based potential). 
%%%%%%%%%%%%%%%%%%%%%%%%%%%%%%%%%%%%%%%%%%%%%%%%%%%%%%%%%%%%%%%%%

%%%%%%%%%%%%%%%%%%%%%%%%%%%%%%%%%%%%%%%%%%%%%%%%%%%%%%%%%%%%%%%%%
\section{Reconstruction Algorithm}\label{sec:algorithm}
We have now the necessary elements to derive the general MAP solution of our reconstruction problem. By using the discrete innovation vector $\ubf$ as an auxiliary variable, we recast the MAP estimation as the constrained optimization problem 
\begin{eqnarray}\label{constrainedMAP}
    \sbf_{{\rm MAP}} = &\arg\underset{\sbf \in \mathbb{R}^K}\min & \left( \frac{1}{2}\|\h\sbf-\y\|_2^2 +\sigma^2
    \sum_{\bk \in \Omega } \Phi_U\left(  u[\bk]  \right) \right)  \nonumber\\ 
    &\centering \text{subject to} & \quad \ubf = \Lbf\sbf\text{.}
\end{eqnarray}
This representation of the solution naturally suggests using the type of splitting-based techniques that have been employed by various authors for solving similar optimization problems~\cite{Wang.etal2008, Ramani.Fessler2010, Afonso.etal2011}. Rather than dealing with a constrained optimization problem directly, we prefer to formulate an equivalent unconstrained problem. To that purpose, we rely on the augmented-Lagrangian method~\cite{Nocedal.Wright2006} and introduce the corresponding \textit{augmented Lagrangian} (AL) functional of~\eqref{constrainedMAP} given by
\begin{eqnarray*}
    \mathcal{L}_\mathcal{A}(\sbf,\ubf,\aalpha) & = &\frac{1}{2}\|\h\sbf-\y\|_2^2 +\sigma^2 \sum_{\bk \in \Omega } \Phi_U \left( u[\bk] \right) \\
    &\mbox{}& +\aalpha^\mathrm{T}(\Lbf\sbf-\ubf) + \frac{\mu}{2} \|\Lbf\sbf-\ubf\|_2^2\text{,}
\end{eqnarray*}
where $\aalpha \in \mathbb{R}^K$ denotes the \textit{Lagrange-multiplier} vector and $\mu \in \mathbb{R}$ is called the \textit{penalty parameter}. The resulting optimization problem takes of the form
\begin{equation}\label{unconstrainedMAP}
\underset{\left( \sbf\in\mathbb{R}^K,~\ubf \in\mathbb{R}^K \right) }{\operatorname{min}}~\mathcal{L}_{\mathcal{A}}(\sbf,\ubf,\aalpha)\text{.}
\end{equation}
To obtain the solution, we apply the \textit{alternating-direction method of multipliers} (ADMM) ~\cite{Boyd.etal2011} that replaces  the joint minimization of the AL functional over $(\sbf, \ubf)$ by the partial minimization of $\mathcal{L}_\mathcal{A}(\sbf,\ubf,\aalpha)$ with respect to each independent variable in turn, while keeping the other variable fixed. These independent minimizations are followed by an update of the Lagrange multiplier. In summary, the algorithm results in the following scheme at iteration $t$:
\begin{subequations}
	\begin{align}
   	 	\qquad \ubf^{t+1} & \leftarrow \arg\underset{\ubf}\min \quad
    		\mathcal{L}_{\mathcal{A}}(\sbf^t,\ubf,\aalpha^t)
		\label{uUpdate}	\\
    		\qquad \sbf^{t+1} & \leftarrow \arg\underset{\sbf}\min \quad
    		\mathcal{L}_{\mathcal{A}}(\sbf,\ubf^{t+1},\aalpha^t)
		\label{xUpdate}	\\
    		\aalpha^{t+1}& = \aalpha^t + \mu(\Lbf\sbf^{t+1}-\ubf^{t+1})\text{.}
		\label{alphaUpdate}	
	\end{align}
\end{subequations}
From the Lagrangian duality point of view, \eqref{alphaUpdate} can be interpreted as the maximization of the dual functional so that, as the above scheme proceeds, feasibility is imposed~\cite{Boyd.etal2011}. 

%=== BEGIN FIGURE: IMAGE SET===% 
\begin{figure*}[th]
\centering
\subfigure[]{
\includegraphics[scale=0.25]{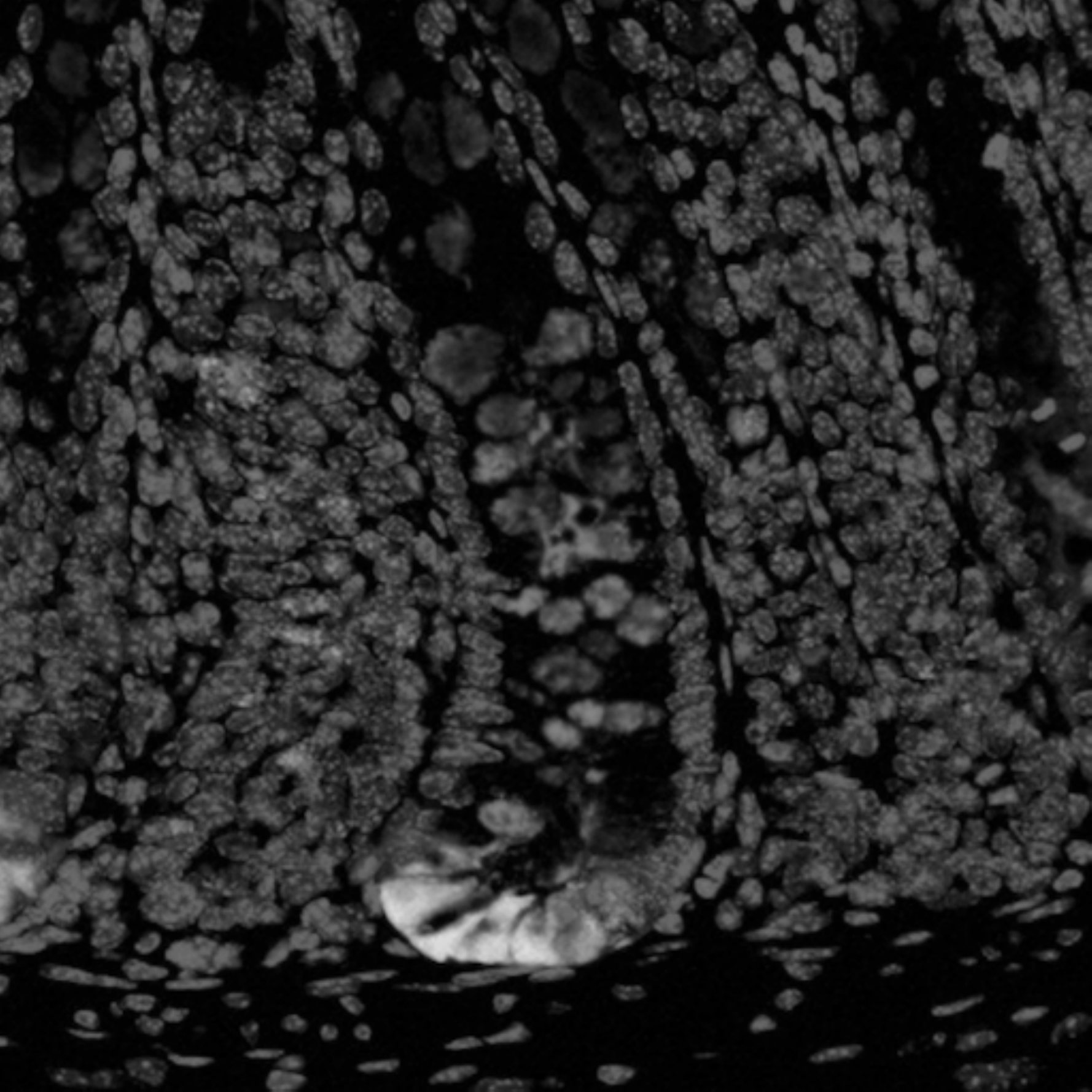}
\label{fig:stemcells}
}
\subfigure[ ]{
\includegraphics[scale=0.25]{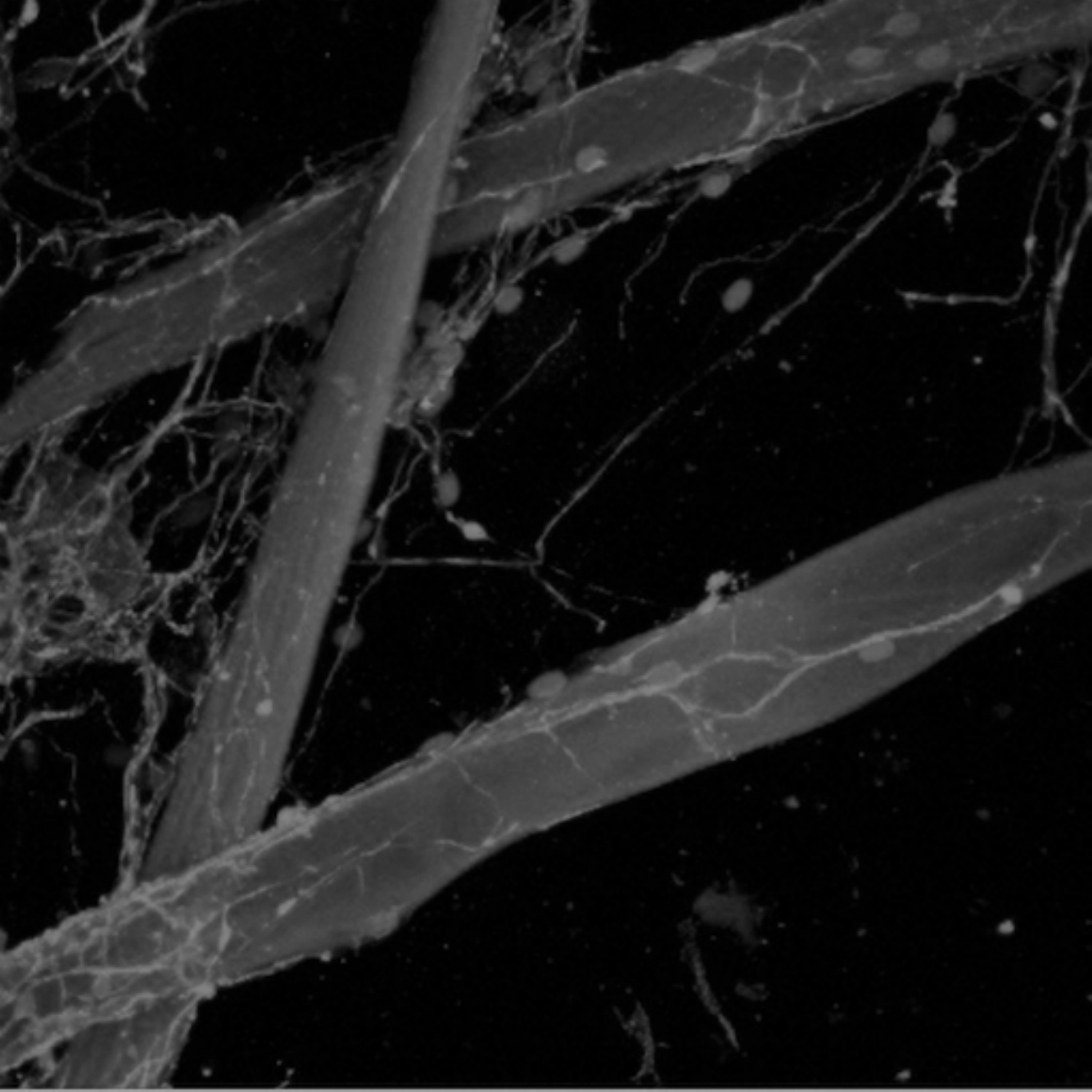}
\label{fig:schwanncells}
}
\subfigure[ ]{
\includegraphics[scale=0.25]{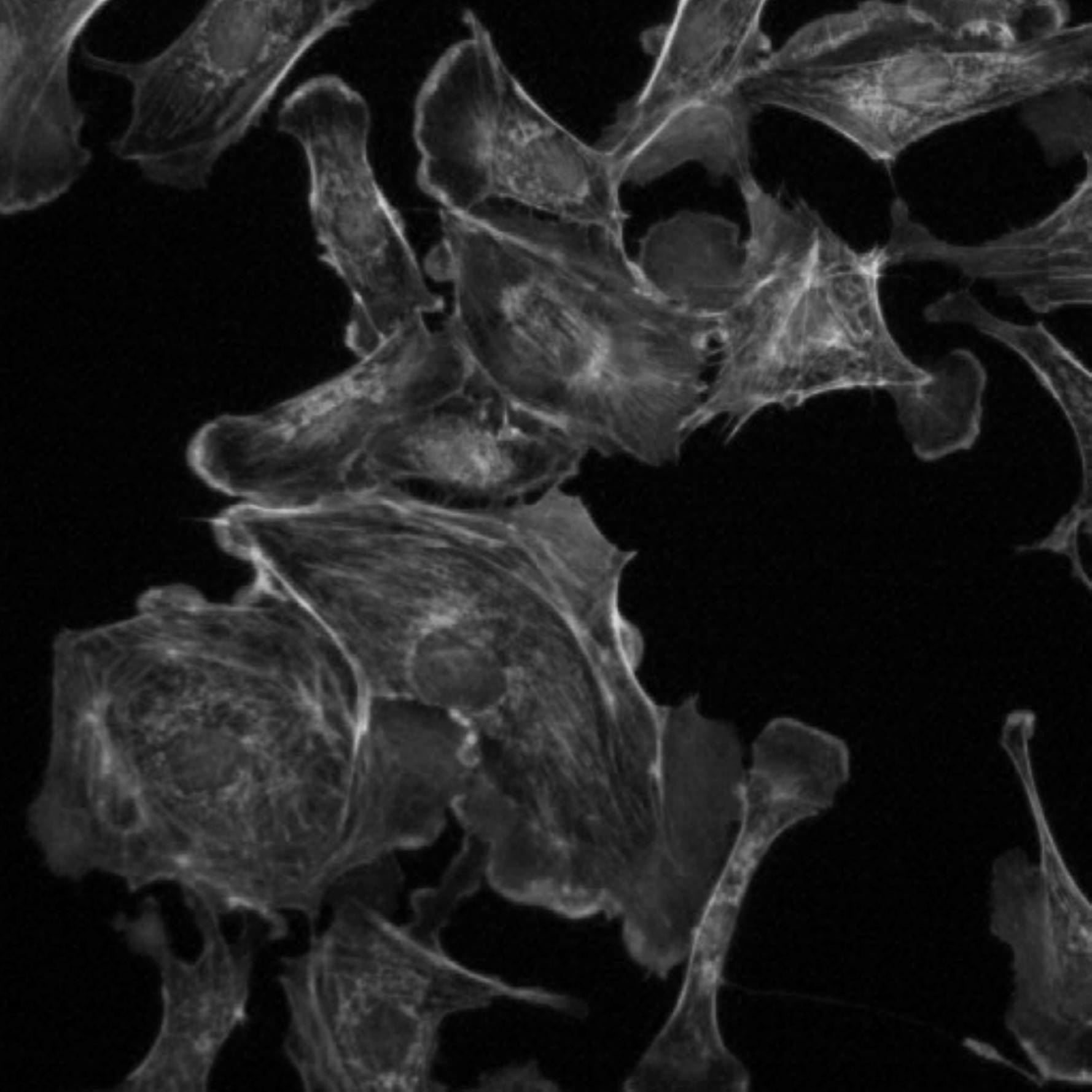}
\label{fig:arterycells}
}
\caption[Optional caption for list of figures]{Images used in deconvolution experiments: (a) stem cells surrounded by goblet cells; (b) nerve cells growing around fibers; (c) artery cells.}
\label{fig:datasetDECONV}
\end{figure*}
%=== END FIGURE: IMAGE SET===%

%=== BEGIN TABLE ===
\begin{table*}
\caption{Deconvolution performance of MAP estimators based on different prior distributions.}
\begin{center}
     \begin{tabular*}{0.53\textwidth}[th]{  r | c c  c  c}
    \hline\hline
                             & BSNR (dB)  & Gaussian  & Laplace      & Student's     \\  \hline
    Stem cells      & 20                 & \bf{14.43}  &    13.76       &  11.86             \\ 

    Stem cells      & 30                & \bf{15.92}  &    15.77       &   13.15            \\ 

    Stem cells      & 40               & \bf{18.11}  &   \bf{18.11} &   13.83          \\ 
    \hline
    Nerve cells & 20            & 13.86     &  \bf{15.31}   &   14.01                \\ 

    Nerve cells & 30           & 15.89     &  \bf{18.18}   &   15.81             \\ 

    Nerve cells & 40           & 18.58    &  \bf{20.57}   &   16.92               \\     
    \hline
    Artery cells & 20                 & 14.86     &  \bf{15.23}   &   13.48                \\ 

    Artery cells & 30                 & 16.59    &  \bf{17.21}   &   14.92                \\ 

    Artery cells & 40                & 18.68    &  \bf{19.61}   &   15.94             \\   
    \hline\hline
  \end{tabular*}
  \end{center}
  \label{tab:res_deconv}  
\end{table*}
%=== END TABLE ===

Now, we focus on the sub--problem~\eqref{uUpdate}. In effect, we see that the minimization is separable, which implies that~\eqref{uUpdate} reduces to performing $K$ scalar minimizations of the form
\begin{equation}\label{eq:scalarMinimization}
    \underset{u[\bk] \in \R} \min ~\left(\sigma^2 \Phi_U(u[\bk] ) + \frac{\mu}{2} \left( u[\bk] - z[\bk] \right)^2 \right)~\text{,}~\forall \bk \in \Omega\text{, }
\end{equation}
where $\zbf= \Lbf\sbf + \aalpha / \mu $.  One sees that~\eqref{eq:scalarMinimization} is nothing but the proximity operator of $\Phi_U(\cdot)$ that is defined below.
\begin{defn}
The proximity operator associated to the function $\lambda \Phi_U(\cdot)$ with $\lambda\in\mathbb{R}_+$ is defined as 
    \begin{equation}
        {\rm prox}_{\Phi_U}(y;\lambda) =
        \arg\underset{x\in \R}\min~\frac{1}{2}(y-x)^2+\lambda\Phi_U(x)\text{.}
    \end{equation}
\end{defn}
Consequently,~\eqref{uUpdate} is obtained by applying ${\rm prox}_{\Phi_U}(z;~\frac{\sigma^2}{\mu})$ in a \textit{component-wise} fashion to
$\zbf= \Lbf\sbf^t+\aalpha^t / \mu$. The closed-form solutions for the proximity operator are well-known for the Gaussian and Laplace priors. They are  given by
\begin{subequations}
	\begin{align}
	&{\rm prox}_{(\cdot)^2}\left(z; \lambda \right) =  z (1+2\lambda)^{-1}\text{,}\\
	&{\rm prox}_{|\cdot|}\left(z; \lambda\right) =  \max(|z| - \lambda, 0){\rm sgn}(z){,}
	\end{align}
\end{subequations}
respectively. The proximity operator has no closed-form solution for the Student's potential.  For this case, we propose to precompute and store it in a lookup table (LUT) (cf. Figure~\ref{fig:prox}). This idea suggests a very fast implementation of the proximal step which is applicable to the entire class of i.d. potentials.

We now consider the second minimization problem~\eqref{xUpdate}, which amounts to the minimization of a quadratic problem for which the solution is given by 
\begin{equation}\label{quadraticProblem}
    \sbf^{t+1}=(\h^{\mathrm{T}}\h+\mu\Lbf^\mathrm{T}\Lbf)^{-1}\left(\h^\mathrm{T}\y+ \mu{\bf L}^\mathrm{T} \left(  {\bf u}^{t+1} - \frac{\boldsymbol \alpha^t }{\mu}\right)\right)\text{.}
\end{equation}
Interestingly, this part of the reconstruction algorithm is equivalent to the Gaussian solution given in~\eqref{eq:solutionTikhonov} and~\eqref{eq:gaussianMAP}.
In general, this problem can be solved iteratively using a linear solver such as the conjugate-gradient (CG) method. Also in some cases, the direct inversion is possible. For instance, when $\h^{\mathrm{T}}\h$ has a convolution structure, as in some of our series of experiments, the direct solution can be obtained by using the FFT~\cite{Hansen.etal2006}.

We conclude this section with some remarks regarding the optimization algorithm. We first note that the method remains applicable when $\Phi_U(x)$ is nonconvex, with the following caveat: When the ADMM converges and $\Phi_U$ is nonconvex, it converges to a local minimum,  including the case where the sub-minimization problems are solved exactly~\cite{Boyd.etal2011}. As the potential functions considered in the present context are closed and proper, we stress the fact that if $\Phi_U:\mathbb{R}\rightarrow \mathbb{R}_+$ is convex and the unaugmented Lagrangian functional has a saddle point, then the constraint in~\eqref{constrainedMAP} is satisfied and the objective functional reaches the optimal value as $t\to\infty$~\cite{Boyd.etal2011}.  Meanwhile, in the case of a nonconvex problems, the algorithm can potentially get trapped in a local minimum in the very early stages of the optimization. It is therefore recommended to apply a deterministic continuation method or to consider a \textit{warm start} that can be obtained by solving the problem first with Gaussian or Laplace priors. We have opted for the latter solution as an effective remedy for convergence issues.

%%%%%%%%%%%%%%%%%%%%%%%%%%%%%%%%%%%%%%%%%%%%%%%%%%%%%%%%%%%%%%%%%

%%%%%%%%%%%%%%%%%%%%%%%%%%%%%%%%%%%%%%%%%%%%%%%%%%%%%%%%%%%%%%%%%
\section{Numerical Results}\label{sec:numericalResults}

In the sequel, we illustrate our method with some concrete examples. We concentrate on three different imaging modalities and consider the problems of deconvolution, MR image reconstruction from partial Fourier coefficients, and image reconstruction from X-ray tomograms. For each of these problems, we present how the discretization paradigm is applied. In addition, our aim is to show that the adequacy of a given potential function is dependent upon the type of image being considered. Thus, for a fixed imaging modality, we perform model-based image reconstruction, where we highlight images that suit well to a particular estimator. For all the experiments, we choose $\Lbf$ to be the discrete-gradient operator. As a result, we update the proximal operators in Section~\ref{sec:algorithm} to their vectorial counterparts. The regularization parameters are optimized via an oracle to obtain the highest-possible SNR. The reconstruction is initialized in a systematic fashion: The solution of the Gaussian estimator is used as initial solution for the Laplace estimator and the solution of the Laplace estimator is used as initial solution for Student's estimator. The $\epsilon$ parameter for Student's estimator is set to $10^{-2}$.
%In this section, we provide simulation results for a variety of linear inverse problems. We consider three imaging modalities: deconvolution microscopy, MR image reconstruction from partial Fourier coefficients and image reconstruction from X-ray tomograms. For all of the experiments, we choose operator $\Lbf$ to be the gradient operator and adopt the vectorial counterparts of the proximal operators in the implementation. The reconstruction quality is measured by SNR.

%In the experiments, the regularization parameters are optimized using an oracle in order to obtain highest possible SNR value for each estimator. We also consider a systematic initialization procedure: the solution of Gaussian estimator is used as the initial solution for Laplace estimator and the one of Laplace estimator is used as the initial solution for Student's estimator. The $\epsilon$ parameter for the Student's estimator is set to $10^{-2}$.

%What we basically aim in the experiments is to show that different estimators overperforms other estimators depending on the underlying image. Rather than trying to promote one particular type of model, we have worked on finding images that are better suited to certain type of regularizations than others while all other factors being the same. 

%=== BEGIN FIGURE: IMAGE SET===% 
\begin{figure*}[th]
\centering
\subfigure[]{
\includegraphics[scale=0.5]{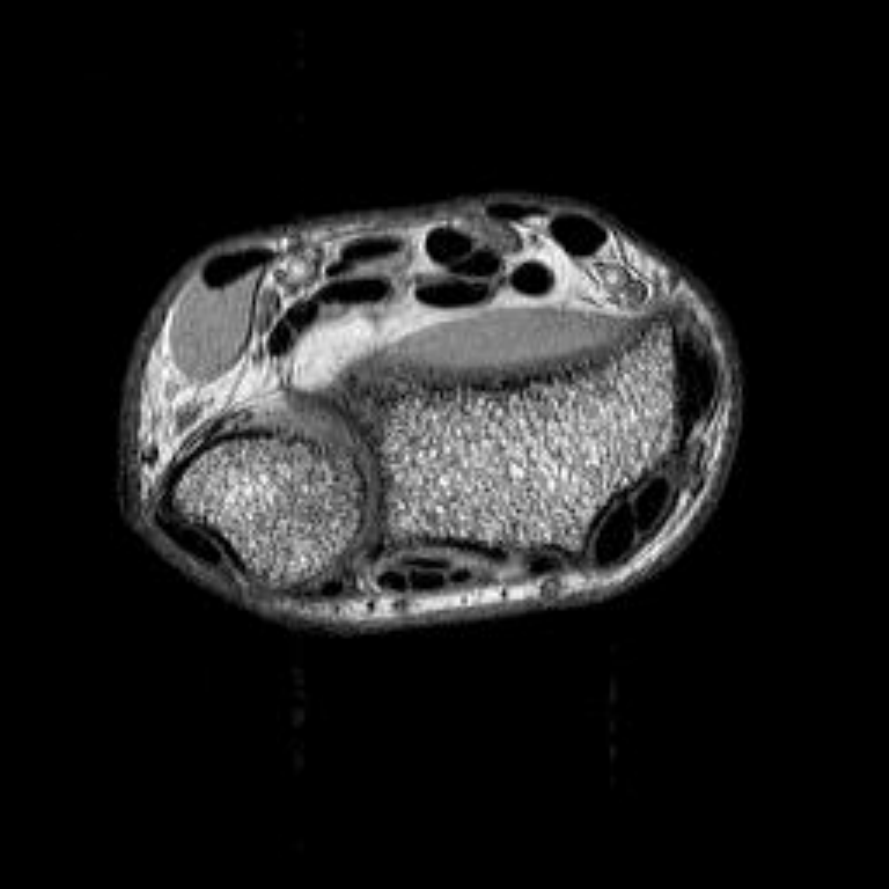}
\label{fig:leg_mri}
}
\subfigure[ ]{
\includegraphics[scale=2]{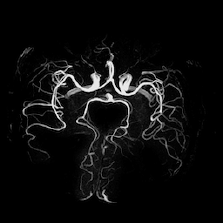}
\label{fig:angio_mri}
}
\subfigure[ ]{
\includegraphics[scale=0.5]{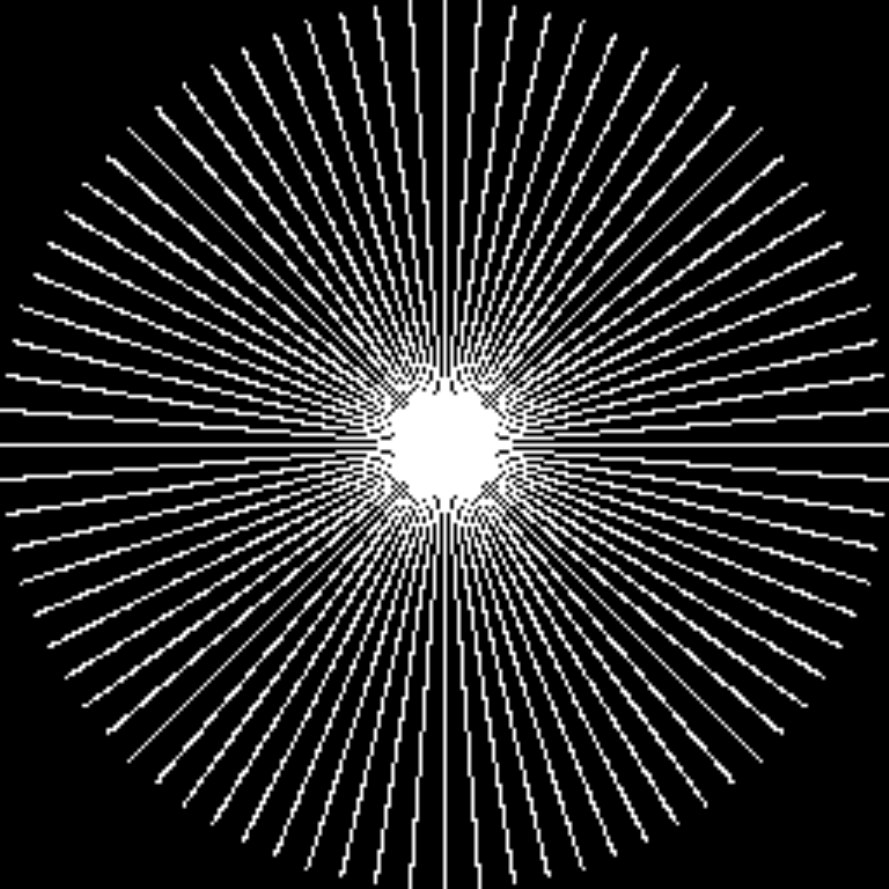}
\label{fig:mask_mri}
}
\caption[Optional caption for list of figures]{Data used in MR reconstruction experiments: (a) cross section of a wrist;  (b) angiography image; (c) k-space sampling pattern along 40 radial lines.}
\label{fig:datasetMRI}
\end{figure*}
%=== END FIGURE: IMAGE SET===% 

%=== BEGIN TABLE ===
\begin{table*}
  \caption{MR image reconstruction performance of MAP estimators based on different prior distributions.}
\begin{center}
     \begin{tabular*}{0.52\textwidth}[th]{  r | c  c  c}
    \hline\hline
                                                                    & Gaussian & Laplace & Student's          \\  \hline
    Wrist (20 radial lines)       &  8.82  &    \bf{11.8} &   5.97                \\ 

    Wrist (40 radial lines)       & 11.30   &    \bf{14.69} &  13.81                \\ 
    \hline
    Angiogram (20 radial lines) & 4.30     &  9.01      &   \bf{9.40}                  \\ 

    Angiogram (40 radial lines) & 6.31     &  14.48   &   \bf{14.97}                \\ 
    \hline\hline
  \end{tabular*}
  \end{center}
  \label{tab:res_mri}  
\end{table*}
%=== END TABLE ===
\subsection{Image Deconvolution}
\label{seq:deconvolution}

The first problem we consider is the deconvolution of microscopy images. In deconvolution, the measurement function in~\eqref{equ:ContinuousMeasurement} corresponds to the shifted version of the point-spread function (PSF) of the microscope on the sampling grid: $\psi^{\rm D}_m(\bx) = \psi^{\rm D}(\bx-\bx_m)$ where $\psi^{\rm D}$ represents the PSF. We discretize the model by choosing $\varphi_{\rm int}(\bx)={\rm sinc}(\bx)$ with $\varphi_\bk(\bx)=\varphi_{\rm int}(\bx-\bx_\bk)$. The entries of the resulting system matrix $\h$ are given by 
\begin{equation}\label{eq:systemMatrixDeconvolution}
\left[\h\right]_{m, \bk} = \langle {\psi^{\rm D}_m(\cdot),\rm sinc}(\cdot-\bx_\bk)\rangle\text{,}
\end{equation}
In effect,~\eqref{eq:systemMatrixDeconvolution} corresponds to the samples of the band-limited version of the PSF. 

We perform controlled experiments, where the blurring of the microscope is simulated by a Gaussian PSF kernel of support $9\times9$ and standard deviation $\sigma_{\rm b}=4$, on three microscopic images of size $512\times512$ that are displayed in Figure~\ref{fig:datasetDECONV}. In Figure~\ref{fig:stemcells}, we show stem cells surrounded by numerous goblet cells. In Figure~\ref{fig:schwanncells}, we illustrate nerve cells growing along fibers, and we show in Figure~\ref{fig:arterycells} bovine pulmonary artery cells.  

For deconvolution, the algorithm is run for a maximum of 500 iterations, or until the relative error between the successive iterates  is less than $5\times 10^{-6}$. Since $\h$ is block-Toeplitz, it can be diagonalized by a Fourier transform under suitable boundary conditions. Therefore, we use a direct FFT-based solver for~\eqref{quadraticProblem}. The results are summarized in Table~\ref{tab:res_deconv}, where we compare the performance of three regularizers for the different blurred SNR (BSNR) levels defined as
$\text{BSNR}={\rm var}(\h\sbf)/\sigma^2$.

We conclude from the results of Table~\ref{tab:res_deconv} that the MAP estimator based on a Laplace prior yields the best performance for images having sharp edges with a moderate amount of texture, such as those in  Figures~\ref{fig:schwanncells}-\ref{fig:arterycells}. This confirms the observation that, by promoting solutions with sparse gradient, it is possible to improve the deconvolution performance. However, enforcing sparsity too heavily, as is the case for Student's priors, results in a degradation of the deconvolution performance for the biological images considered. Finally, for a heavily textured image like the one found in Figure~\ref{fig:stemcells}, image deconvolution based on Gaussian priors yields the best performance. We note that the derived algorithms are compatible with the methods commonly used in the field (e.g.,\ Tikhonov regularization~\cite{Preza.etal1993} and TV regularization~\cite{Dey.etal2006}).

%=== BEGIN FIGURE: IMAGE SET===% 
\begin{figure*}[th]
\centering
\subfigure[]{
\includegraphics[scale=0.44]{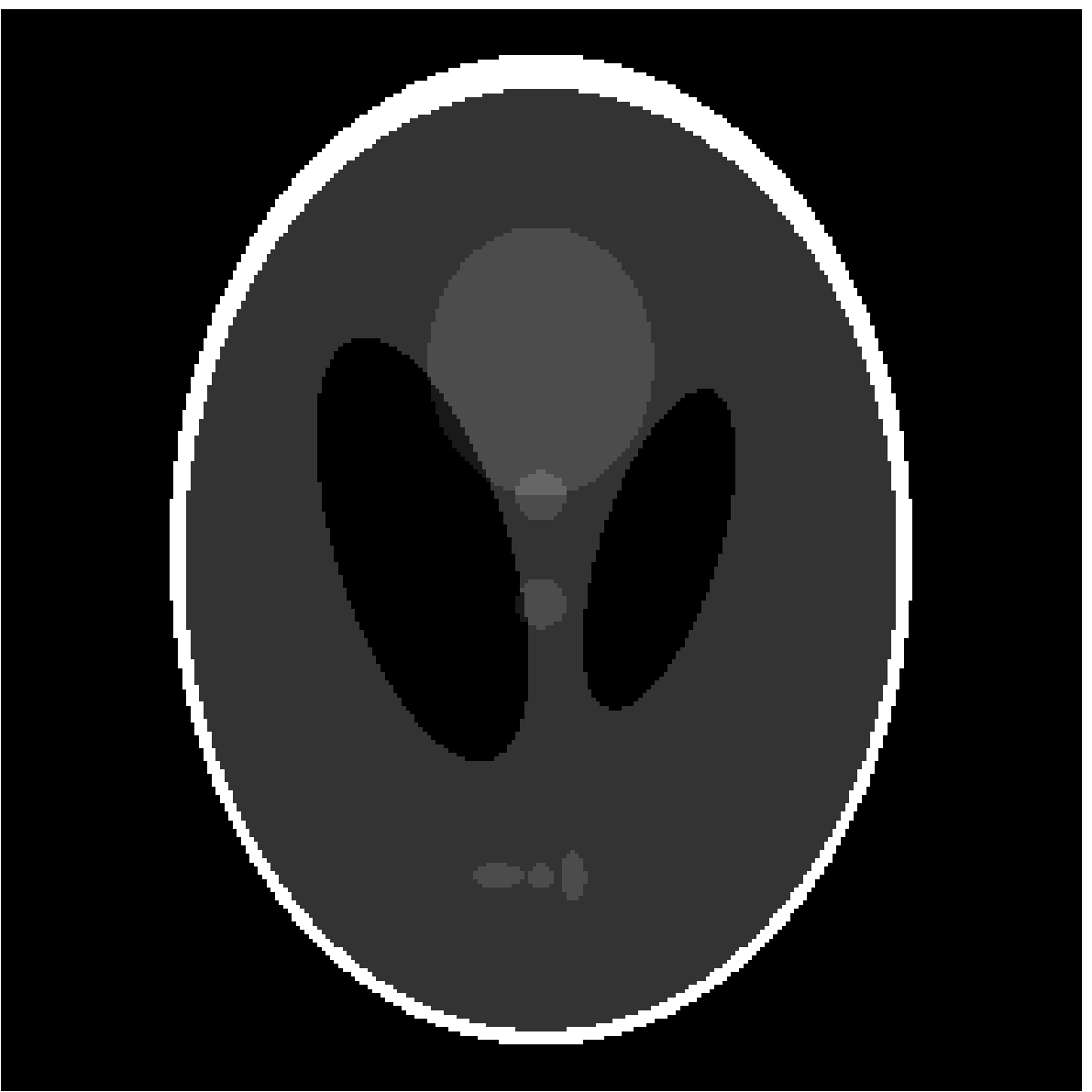}
\label{fig:subfig1_xray}
}
\subfigure[ ]{
\includegraphics[scale=0.20]{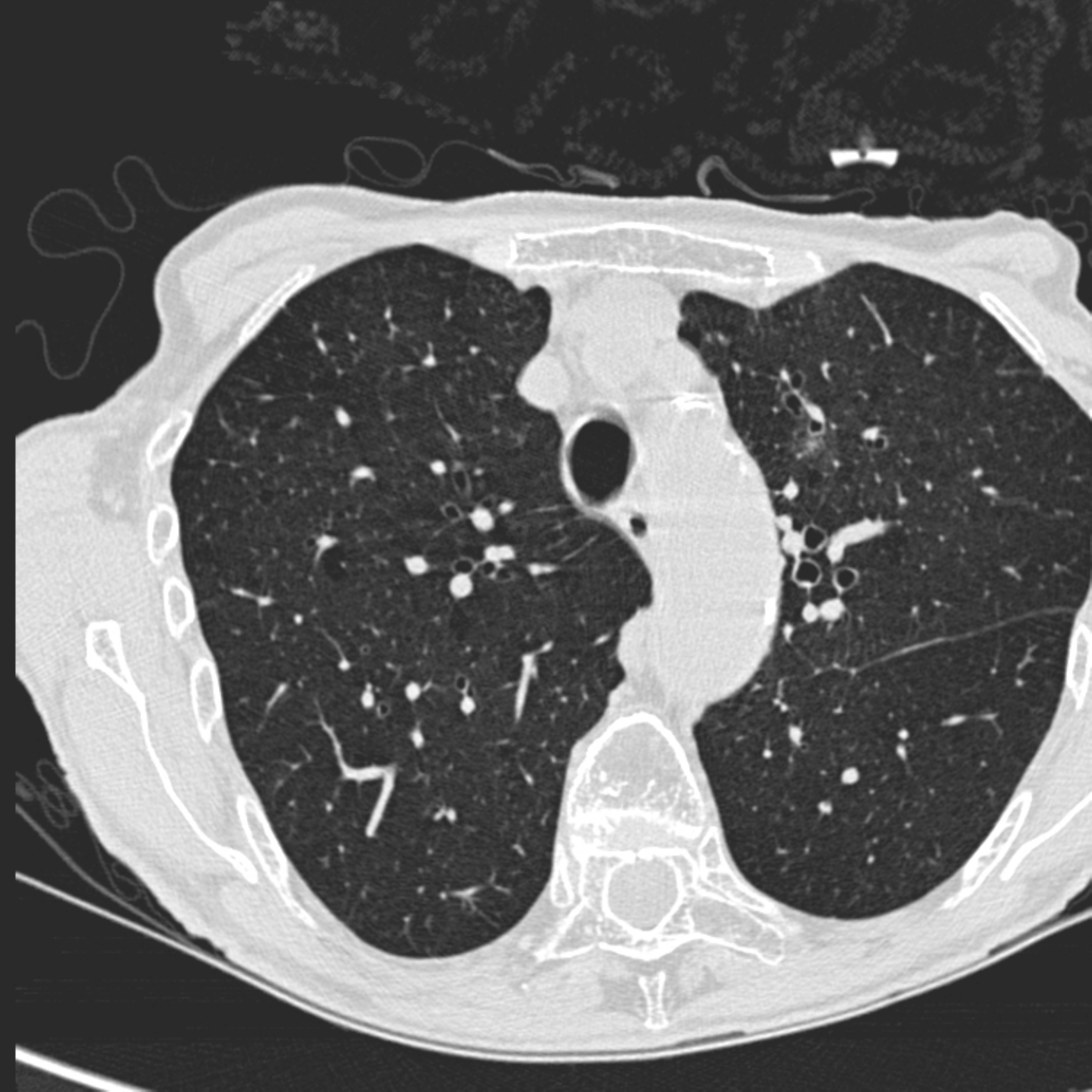}
\label{fig:subfig2-xray}
}
\caption[Optional caption for list of figures]{Images used in X-ray tomographic reconstruction experiments: (a) the Shepp-Logan (SL) phantom; (b) cross section of the lung.    }
\label{fig:datasetXRAY}
\end{figure*}
%=== END FIGURE: IMAGE SET===% 

%=== BEGIN TABLE ===
\begin{table*}
  \caption{Reconstruction results of X-ray computed tomography using different estimators.}
\begin{center}
     \begin{tabular*}{0.53\textwidth}[th]{  r | c c  c}
    \hline\hline
                             & Gaussian & Laplace & Student's \\  \hline
          SL Phantom (120 direction)& 16.8         &  17.53  &   \bf{18.76}                \\ 

        SL Phantom (180 direction)& 18.13         &  18.75  &   \bf{20.34}                \\ 

            \hline     
                Lung (180 direction)& \bf{22.49} &    21.52      &   21.45                \\ 
 
    Lung (360 direction) & \bf{24.38} &    22.47      &   22.37                \\

    \hline\hline
  \end{tabular*}
  \end{center}
  \label{tab:res_xray}  
\end{table*}
%=== END TABLE ===

\subsection{MRI Reconstruction}
We consider the problem of reconstructing MR images from undersampled $\bk$-space trajectories. The measurement function represents a complex exponential at some fixed frequencies and is defined as $\psi_m^{\rm M}(\bx)={\rm e}^{2\pi {\rm j}\langle\bk_m,\bx \rangle}$ where $\bk_m$ represents the sample point in $\bk$-space. As in Section~\ref{seq:deconvolution}, we choose $\varphi_{\rm int}(\bx)={\rm sinc}(\bx)$ for the discretization of the forward model, which results in a system matrix with the entries 
\begin{align}
\left[\h\right]_{m, \bk} &= \langle \psi_m^{\rm M}(\bx),{\rm sinc}(\cdot-\bx_\bk)\rangle\nonumber \\
                                     &= {\rm e}^{-{\rm j}2\pi \langle\bk_m,\bx_\bk \rangle} \,\,\,\, \text{if} \,\,\,\, |\bk_m|_{\infty} \leq \half\text{.}
\end{align}
The effect of choosing a ${\rm sinc}$ function is that the system matrix reduces to the discrete version of complex Fourier exponentials.

We study the reconstruction of the two MR images of size $256\times256$ illustrated Figure~\ref{fig:datasetMRI}---a cross-section of a wrist is displayed in the first image, followed by an MR angiography image---and consider a radial sampling pattern in $\bk$-space (cf.\ Figure~\ref{fig:mask_mri}).  

The reconstruction algorithm is run with the stopping criteria set as in Section~\ref{seq:deconvolution} and an FFT-based solver is used for~\eqref{quadraticProblem}. We show in Table~\ref{tab:res_mri} the reconstruction performance of our estimators for different number of radial lines. 

On one hand, the estimator based on Laplace priors yield the best solution in the case of the wrist image, which has sharp edges and some amount of texture. Meanwhile, the reconstructions using Student's priors are suboptimal because they are too sparse. This is similar to what was observed with microscopic images. On the other hand, Student's priors are quite suitable for reconstructing the angiogram, which is mostly composed of piecewise-smooth components. We also observe that the performance of Gaussian estimators is not competitive for the images considered. Our reconstruction algorithms are tightly linked with the deterministic approaches used for MRI reconstruction including TV~\cite{Block.etal2007} and log-based reconstructions~\cite{Trzasko.Manduca2009}.

\subsection{X-Ray Tomographic Reconstruction}

X-ray computed tomography (CT) aims at reconstructing an object from its projections taken along different directions.  The mathematical model of a conventional CT is based on the Radon transform
\begin{align*}
g_{\theta_m}(t_m)&={\cal{R}}_{\theta_m}\{s(\bx)\}(t_m)\\
               &=\int_{\R^2} s(\bx)\delta(t_m-\langle \bx,\boldsymbol{\theta_m}\rangle){\rm d}\bx\,,
\end{align*} 
where $s(\bx)$ is the absorption coefficient distribution of the underlying object, $t_m$ is the sampling point and $\boldsymbol{\theta}_m=({\cos}(\theta_m), {\sin}(\theta_m))$ is the angular parameter.  Therefore, the measurement function $\psi_m^{\rm X}(\bx)=\delta(t_m-\langle \bx,\boldsymbol{\theta}_m\rangle)$ denotes an idealized line in $\R^2$  perpendicular to $\boldsymbol{\theta}_m$. In our formulation, we represent the absorption distribution in the space spanned by the tensor product of two B-splines
\begin{equation*}
s({\bf{x}})=\sum_{\bf{k}} s[{\bf{k}}] \varphi_{\rm int}(\bx-\bk)\,,
\end{equation*}   
where $\varphi_{\rm int}(\bx)=\text{tri}(x_1)\text{tri}(x_2)$, with $\text{tri}(x)=\left(1-|x|\right)$ denoting the linear B-spline function. The entries of the system matrix are then determined explicitly using the B-spline calculus described in~\cite{Entezari.etal2012}, which leads to
\begin{align*}
\left[\h\right]_{m, \bk} &= \left\langle \delta(t_m-\langle\bx,\boldsymbol{\theta}_m\rangle), \varphi_{\rm int}(\bx-\bk)\right\rangle \\
                                &= \frac{\bigtriangleup_{\abs{\cos\theta_m}}^{2}\bigtriangleup_{\abs{\sin\theta_m}}^{2}}{3{!}}(t_m-\langle\bk,\boldsymbol{\theta}_m\rangle)_{+}^{3},
\end{align*}
where $\bigtriangleup_{h}f(t)=\frac{f(t)-f(t-h)}{h}$ is the finite-difference operator, $\bigtriangleup_{h}^nf(t)$ is its $n$-fold iteration, and $t_+=\text{max}(0,t)$.
This approach provides an accurate modeling, as demonstrated in~\cite{Entezari.etal2012} where further details regarding the system matrix and its implementation are provided. 

We consider the two images shown in Figure~\ref{fig:datasetXRAY}. The Shepp-Logan (SL) phantom has size $256\times 256$, while the cross section of the lung has size $750\times 750$. In the simulations of the forward model, the Radon transform is computed along 180 and 360 directions for the lung image and along 120 and 180 directions for the SL phantom. The measurements are degraded with the Gaussian noise such that the signal-to-noise ratio is 20 dB.  

For the reconstruction, we solve the quadratic minimization problem~\eqref{quadraticProblem} iteratively by using 50 CG iterations. The reconstruction results are reported in Table~\ref{tab:res_xray}. 

The SL phantom is a piecewise-smooth image with sparse gradient. We observe that the imposition of more sparsity brought by Student's priors significantly improves the reconstruction quality for this particular image. On the other hand, we find that the Gaussian priors for the lung image outperform the other priors. Like the deconvolution and MRI problems, our algorithms are in line with Tikhonov-type~\cite{Wang.etal2006} and TV~\cite{Qun.Jacques2005} reconstructions used for X-ray CT.

\subsection{Discussion}
As our experiments on different types of imaging modalities have revealed, sparstity-promoting reconstructions are powerful methods for solving biomedical image reconstruction problems. However, encouraging sparser solutions does not always yield the best reconstruction performance and non-sparse solutions provided by Gaussian priors still yields better reconstructions for certain images. The efficiency of a potential function is primarily dependent upon the type of image being considered. In our model, this is related to the L\'evy exponent of the underlying continuous-domain innovation process $w$ which is in direct relationship with the signal prior. 
%%%%%%%%%%%%%%%%%%%%%%%%%%%%%%%%%%%%%%%%%%%%%%%%%%%%%%%%%%%%%%%%%

%%%%%%%%%%%%%%%%%%%%%%%%%%%%%%%%%%%%%%%%%%%%%%%%%%%%%%%%%%%%%%%%%
\section{Conclusion}
The purpose of this paper has been to develop a practical scheme for linear inverse problems by combining a proper discretization method and the theory of continuous-domain sparse stochastic processes. On the theoretical side, an important implication of our approach is that the potential functions cannot be selected arbitrarily as they are necessarily linked to infinitely divisible distributions. The latter puts restrictions but also constitutes the largest family of distributions that is closed under linear combinations of random variables. On the practical side, we have shown that the MAP estimators based on these prior distributions cover the current state-of-the-art methods in the field including $\ell_1$-type regularizers. The class of said estimators is sufficiently large to reconstruct different types of images. 

%which remarkably cover the current state-of-the-art methods in the field including $\ell_1$-type regularizers. In this sense, the our formalism serves as a tool that is flexible enough to be applicable for different types of images. 

Another interesting observation is that  we face an optimization problem for MAP estimation  that is generally nonconvex, with the exception of the Gaussian and the Laplacian priors. We have proposed a computational solution, based on alternating-direction method of multipliers, that applies to arbitrary potential functions by suitable adaptation of  the proximity operator. 

In particular, we have applied our framework to deconvolution, MRI, and X-ray tomographic reconstruction problems and have compared the reconstruction performance of different estimators corresponding to models of increasing sparsity. 

In basic terms, our model is composed of two fundamental concepts: the whitening operator $\Lop$, which is in connection with the regularization operator, and the L\'evy exponent $f$, which is related to the prior distribution. A further advantage of continuous-domain stochastic modeling is that it enables us to investigate the statistical characterization of the signal in any transform domain. This observation designates key research directions: (1) the identification of the optimal whitening operators and (2) the proper fitting of the L\'evy exponent of the continuous-domain innovation process $w$ to the class of images of interest.

\bibliographystyle{IEEEtran}
\bibliography{references}

\end{document}